\documentclass[12pt]{article}
\usepackage[toc,page]{appendix}
\usepackage{graphicx}
\usepackage{amssymb,amsmath}
\usepackage{epstopdf}
\usepackage{epsf}
\usepackage{graphicx}
\usepackage{setspace}
\usepackage{amsthm}
\usepackage{dsfont}
\usepackage{breqn}
\usepackage{mathrsfs}
\usepackage{booktabs}
\usepackage{setspace}
\usepackage{listings}
\usepackage{multirow}
\usepackage{color}
\usepackage[usenames,dvipsnames,svgnames,table]{xcolor}
\usepackage[margin=1.0in]{geometry}
\usepackage[colorlinks=true,
            linkcolor=red,
            urlcolor=blue,
            citecolor=blue]{hyperref}
\usepackage{rotating}
\usepackage{mathrsfs}
\usepackage{bm}
\usepackage{tabulary}
\usepackage{algorithm}
\usepackage{algorithmic}
\usepackage{natbib}
\usepackage[compact]{titlesec}
\titleformat*{\section}{\bfseries\filcenter}
\titleformat*{\subsection}{\bfseries}
\titleformat*{\subsubsection}{\itshape}
\titlespacing{\section}{0pt}{0pt}{0pt}
\makeatletter
\renewcommand\@biblabel[1]{}

\makeatother
\newcommand*{\TitleFont}{%
      \usefont{\encodingdefault}{\rmdefault}{b}{n}%
      \fontsize{22}{24}%
      \selectfont}
\usepackage{titling}

\newtheorem{corollary}{Corollary}
\newtheorem{proposition}{Proposition}
\newtheorem{lemma}{Lemma}
\newtheorem{theorem}{Theorem}
\def \pwp1{\stackrel{a.s.}{\to}}

\def \bx{\bm x}
\def \bbeta{\bm \beta}
\def \tbeta{\tilde{\beta}}

\def \bg{\begin{equation}}
\def \ed{\end{equation}}
\def \by{\bm y}

\def \bvarphi{\bm \varphi}

\def\GP{{\mathcal {GP}}}

\def\k{\kappa}
\def\iid{\stackrel{\mathrm{iid}}{\sim}}

\def \bx{\bm x}
\def \bbeta{\bm \beta}
\def \tbeta{\tilde{\beta}}

\def \bs {\bm s}
\def \bg{\begin{equation}}
\def \ed{\end{equation}}
\def \by{\bm y}

\def \bvarphi{\bm \varphi}

\def\GP{{\mathcal {GP}}}

\def\k{\kappa}
\def\iid{\stackrel{\mathrm{iid}}{\sim}}
\def\mS{\mathcal{S}}

\def\mC{\mathcal{C}}

\def\mR{\mathcal{R}}
\def\TMGP{{\mathcal {TMGP}}}

\def\mP{\mathcal P}
\def\mR{\mathcal{R}}
\def\bP{\mathbb{P}}
\def\bX{\bm X}
\def\rT{\scriptsize{\mathrm{T}}}
\def\balpha{\boldsymbol{\alpha}}

\title{\TitleFont Thresholded Multiscale Gaussian Processes with Application to Bayesian Feature Selection for Massive Neuroimaging Data}
\date{}
\author{Ran Shi\thanks{Ran Shi is Ph.D. Student, Department of Biostatistics and Bioinformatics, Emory University, Atlanta, GA 30322., Jian Kang is Assistant Professor, Department of Biostatistics and Bioinformatics and Department of Radiology and Imaging Sciences, Emory University, Atlanta, GA 30322.}\ \  and\ \  Jian Kang\thanks{To whom correspondence should be addressed: jian.kang@emory.edu}
}

\begin{document}
\maketitle
\setlength{\belowdisplayskip}{1.5pt} \setlength{\belowdisplayshortskip}{1.5pt}
\setlength{\abovedisplayskip}{1.5pt} \setlength{\abovedisplayshortskip}{1.5pt}
\begin{abstract}
Motivated by the needs of selecting important features for massive neuroimaging data, we propose a spatially varying coefficient model (SVCMs) with sparsity and piecewise smoothness imposed on the coefficient functions.  A new class of nonparametric priors is developed based on thresholded multiscale Gaussian processes (TMGP).  We show that the TMGP has a large support on a space of sparse and piecewise smooth functions, leading to posterior consistency in coefficient function estimation and feature selection. Also, we develop a method for prior specifications of thresholding parameters in TMGPs. Efficient posterior computation algorithms are developed by adopting a kernel convolution approach, where a modified square exponential kernel is chosen taking the advantage that the analytical form of the eigen decomposition is available. Based on simulation studies, we demonstrate that our methods can achieve better performance in estimating the spatially varying coefficient. Also, the proposed model has been applied to an analysis of resting state functional magnetic resonance imaging (Rs-fMRI) data from the Autism Brain Imaging Data Exchange (ABIDE) study, it provides biologically meaningful results.
\end{abstract}
\clearpage
\section{Introduction}
\label{Introduction}
Recent advancements in biomedical imaging technologies have provided abundant information and extensive resources for researchers to learn the human brain and neurological diseases. A variety of imaging modalities, such as magnetic resonance imaging (MRI), diffusion tensor imaging (DTI), functional magnetic resonance imaging (fMRI) and positron emission tomography (PET)  have been developed to measure brain structures and functions from different perspectives, generating various large-scale  spatially distributed measurements  over a three dimensional (3D) space of the human brain. We refer to those massive spatial measurements of brain as neuroimages.   This poses great opportunities and new challenges for neuroscientists and statisticians to develop efficient analytical methods that extract useful features from neuroimages to characterize the association between the brain activities and neurological diseases. To this end,  regression analysis,  a general and flexible modeling framework for studying the association among variables,  has been investigated and considered as a powerful tool in the analysis of massive neuroimaging data, where neuroimages are modeled as outcome variables; and the disease status along with the clinical, biological and demographical information can all potentially be predictors. 

 A pioneer work using the regression model for the neuroimaging data is the mass univariate analysis (MUA). This approach fits a general linear model (GLM) at each spatial location in the brain (to which is referred as a voxel) and obtains massive test statistics over space to identify voxels/regions that are significantly associated with a specific covariate, which requires multiple comparisons correction. One standard procedure  is to calculate the family-wise error rate (FWER) based on the random field theory for statistical parametric maps \citep{friston1995statistical, worsley04, lazar08}. Another approach is to control the false discovery rate (FDR) using the observed p-values~\citep{benjamini2001control, genovese2002thresholding}. A major drawback of MUA is that the models do not borrow information from the spatial dependence across brain locations. In practice, the neuroimaging data are usually pre-processed by a spatial smoothing procedure using a kernel convolution approach. Performing MUA on these pre-smoothed data may lead to inaccuracy and low efficiency in terms of estimating and testing the covariate effects \citep{chumbley09, li2011multiscale}. Recent development in adaptive smoothing methods for preprocessing \citep{yue2010adaptive} and estimation \citep{polzehl2000adaptive, qiu2007jump, tabelow2008accurate, tabelow2008diffusion, li2011multiscale, wang2013multiscale} may improve the performance in terms of reducing noise and preserving features. It is especially powerful to detect delicate features such as jump discontinuities, which is one of the universal characteristics for neuroimaging data \citep{chan2005image, tabelow2008diffusion, tabelow2008accurate, chumbley09}.  

To achieve a similar goal in the analysis of neuroimaging data, \citet{zhu2014spatially} recently developed a systematic modeling approach using a novel spatially varying coefficient model (SVCM) which incorporates both spatial dependence and piecewise smooth covariate effects. General SVCMs have been extensively investigated and developed for different applications in environmental heath, epidemiology, ecology and geographical studies as demonstrated in \citet{cressie1993statistics, diggle1998model, gelfand2003spatial, smith2002predicting}. The SVCM encompasses a wide range of regression models with the outcome variable observed over space and the regression coefficients modeled as functions varying spatially. We refer to this type regression coefficients as spatially varying coefficient functions (SVCFs). SVCFs are commonly assumed to be  smooth functions or $\rho$ times continuously differentiable functions with $\rho\geq 1$ (we will not make this distinction throughout the rest of this paper unless noted). To model to SVCFs, smooth spatial processes are usually employed to characterize the dependence structure over space.  A pure noise process is typically introduced to capture random variabilities, i.e. the ``spatial nugget effect". \citet{zhu2014spatially} extended the general SVCMs by introducing jump discontinuities into the SVCFs, making the model especially useful for neuroimaging data analysis. Based on stepwise multiscale estimating procedures and asymptotic Wald tests, \citet{zhu2014spatially}'s SVCM also can identify the brain regions that are significantly associated with the given covariates, although it is not developed particularly for feature selection.

In this article, we aim to develop a Bayesian feature selection method for large-scale neuroimaging data that can directly select imaging features associated with covariates of interest.  Regularization methods have been studied extensively for variable selection in regression models \citep{tibshirani1996regression, fan2001variable, zou2006adaptive, candes2007dantzig}. Bayesian methods have also been developed based on various prior specifications.  \citet{mitchell1988bayesian} developed a prior model for linear model coefficients using the mixture of a uniform distribution (slab) and a point mass at zero (spike), which is broadly referred to as the spike-and-slab type of priors. \citet{george1993variable} proposed to use the scale mixture of two zero-mean Gaussian distributions and developed posterior computation algorithm based on Gibbs sampling. Relative works also include but not limited to \citet{ishwaran2005spike, liang2008mixtures, park2008bayesian, hans2009bayesian, bondell2012consistent,johnson2012bayesian,armagan2013generalized,polson2013bayesian, narisetty2014bayesian}. Most of these priors were initially introduced for independent regression coefficients. In light of the needs of integrating complex data structure in many applications,  recent development of Bayesian variable selection incorporates dependence structures into the prior model.  \citet{li2010bayesian} assumed that covariates lay on an undirected graph and used the Ising prior to incorporate this information to the model space and applied this method to analyze the genomics data. For the modeling of spatial data, Markov random field (MRF) is one of the commonly used priors for regression coefficients. For instance, \citet{smith2003assessing, smith2007spatial} applied this type of priors to fMRI data analyses. For the analysis of physical activity and environmental health data, \citet{reich2010bayesian} developed an multivariate SVCM along with a Bayesian variable selection procedure to identify important SVCFs, using the spike-and-slab prior. Their focus, however, was on distinguishing covariate effects that were zero constant, nonzero constant and spatially varying instead of selecting features within the varying coefficient functions.

As the aforementioned variable selection methods cannot be directly applied to identify important features for massive neuroimaging data in the SVCM framework. To fill this gap, we develop a novel Bayesian nonparametric prior model for the SVCF. We refer to it as the thresholded multiscale Gaussian process (TMGP). The TMGP prior is constructed by thresholding a multiscale Gaussian process, which is a combination of two types of GPs: a global GP to account for the entire domain spatial dependence and a local GP to accommodate the regional fluctuations.  
Thus,  the TMGP can characterize important common features of the neuroimaging data, including sparsity, global spatial dependence, piecewise smoothness, edge effects and jumps. The proposed TMGP prior enjoys the large support property, leading to posterior strong consistency in estimation and feature selection for the sparse and piecewise smooth SVCFs. We also develop efficient MCMC posterior computation algorithms based on a kernel convolution approach. A special choice of the kernel function enables the computation scalable to an ultra-high dimensional case. 

The remaining parts of the article is organized as follows: In Section 2, we first introduce the SVCMs for neuroimaging data analysis and particularly discuss conditions on SVCFs in the proposed model. Then we present the construction of TMGP which serves as a prior model for SVCFs. In Section 3, we study the theoretical properties of TMGP and the proposed SVCMs. In Section 4, we develop an efficient and scalable posterior computation algorithm based on a kernel convolution approach. We evaluate the performance of proposed method via simulation studies and and analyze the ABIDE data in Section 5. We conclude our work with a brief discussion on the future work in Section 6. 

\section{Feature selection within the spatially varying coefficient functions}
We start with  general notations and definitions.  Denote by $\mathbb{R}^p$ a $p$-dimensional real Euclidean space for any $p\geq 1$.  For any $\bbeta \in\mathbb{R}^p$, write $\bbeta = (\beta_1,\beta_2,\ldots, \beta_p)^{\rT}$,  define $\|\bbeta\|_\infty=\max_{1\leq k\leq p} |\beta_k|$, $\|\bbeta\|_1=\sum_{k=1}^p |\beta_k|$ and $\|\bbeta\|_2=\sqrt{\sum_{k=1}^p \beta_k^2}$. Denote by $\mR\subset \mathbb{R}^3$ a compact region in the standard brain space. Let $\bs_1,\ldots, \bs_n \in \mR$ be  a set of spatial locations where brain signals are measured. An empirical measure on $\mR$ induced by $\{\bs_1,\ldots, \bs_n\}$ is defined as $\bP_n(d\bs)=\frac{1}{n}\sum_{i=1}^n I[\bs_i\in d\bs]$, where the indicator function $I[\mathcal{A}] = 1$ if event $\mathcal{A}$ occurs, $I[\mathcal{A}] = 0$, otherwise.  For a scalar-valued function $\beta(\cdot): \mR \mapsto \mathbb R$, define $\|\beta(\cdot)\|_\infty=\sup_{\bs\in\mR} |\beta(\bs)|$ and $ \|\beta(\cdot)\|_1 = \int_{\bs\in \mR}|\beta(\bs)|\mathbb{P}_n(d\bs)$. For a $p$-dimensional vector-valued function $\bbeta(\bs)=[\beta_1(\bs),\ldots,\beta_p(\bs)]^{\rT}: \mR \mapsto \mathbb R^p$, define $\|\bbeta(\cdot)\|_{1,\infty}=\max_{1\leq k\leq p}\|\beta_k(\bs)\|_1$. Denote by $\mC(\mR)$ a collection of all the continuous functions defined on $\mR$. Let $D^{\balpha}\beta$ be a partial derivative operator on function $\beta(\cdot)$ (given its existence) which is given by 
$\frac{\partial^{\|\balpha\|_1}\beta}{\partial s_1^{\alpha_1}\cdots \partial s_d^{\alpha_d}}$ for $\balpha\in \mathbb{R}^p$.   Denote by $\mC^\rho(\mR)$ a set of functions $\beta$ defined on $\mR$ with continuous partial derivatives $D^\balpha\beta$ for all $\balpha$ such that $\|\balpha\|_1\leq \rho$.

\subsection{The spatially varying coefficient model for neuroimaging data}
Suppose the data set consists of $m$ subjects. For each subject $j$, let $y_j(\bs)$ be the brain signal measured from a certain imaging modality at location $\bs \in \mR$; and there are also $p$ covariates are collected, denoted $\bx_j=(x_{j1},\ldots, x_{jp})^{\rT}$, for $j = 1,\ldots, m$. The spatially varying coefficient model (SVCM) for neuroimaging data is given by 
\bg \label{model} y_j(\bs) = \bx_{j}^{\rT}\bbeta(\bs) +  e_j(\bs), \ed
where $\bbeta(\bs) = [\beta_1(\bs),\ldots,\beta_p(\bs)]^{\rT}$ is the spatially varying coefficient function (SVCF) defined on $\mR$. It characterizes associations between covariates and imaging outcomes. To be more specific, $\beta_k(\bs)$ ($k=1,\ldots, p$) quantifies the $k$th covariate effects at brain location $\bs$. The independent error process $e_j(\bs)$ is the assumed to be spatially homogeneous across the whole brain for each subject. In our model, we simply assume that $e_j(\bs) \iid N(0, \sigma^2)$ for all subjects at all locations.

For neuroimaging data from commonly used imaging modalities, only outcomes at locations $\bs_1$, $\ldots$, $\bs_n$ are observed. At these locations, the SVCM proposed in \eqref{model} for the recorded neuroimaging data can be expressed as
\bg\label{outcome}[ \by(\bs_i)\mid \bbeta(\bs_i), \sigma^2] \sim N( \bm X \bbeta(\bs_i), \ \sigma^2\bm I_m)\ed
independently for all $i=1,\ldots,n$, where $\by(\bs_i) = [y_1(\bs_i),\ldots,y_m(\bs_i)]^{\rT}$, $\bm X = [\bm x_1,\ldots,\bm x_m]^{\rT}$ and $\bm e(\bs_i) = [e_1(\bs_i),\ldots,e_m(\bs_i)]^{\rT}$. For simplicity, denote by $\mathcal Y=\{\by(\bs_i)\}_{i=1}^n$ an $m\times n$ matrix recoding all the neuroimaging outcomes involved in the study.

In neuroimaging studies, there exists a natural region partition of the whole brain domain $\mR$ into bounded connected sets $\mR_1,...,\mR_{G}$ with non-empty interiors, such that $\mR=\cup_{g=1}^{G} \mR_g, \ \mR_g\cap\mR_{g^\prime}=\emptyset, \ \forall g\neq g^\prime$. In many cases, one can utilize $R_g$ as neuroanatomical regions from commonly used labeling systems such as the Automated Anatomical Labeling (AAL) \citep{tzourio2002automated}. For region of interest (ROI) based analysis, each ROI is one parcellated region. For seed-based region-level analysis, $R_g$ can be regarded as the clusters showing strong functional connectivities based on preliminary results. In some voxelwise analysis with no regional information to be incorporated, we can simply consider each voxel (a 3D cubic) as a region and the centers of voxels as observed brain locations $\bs_1,\ldots, \bs_n$. 

To utilize the proposed SVCM for analysis of neuroimaging data and feature selection, we state our assumptions for the SVCF in model \eqref{model}. Specifically, we work with piecewise smooth ($\rho$-times continuously differentiable to be accurate) functions with structured sparsity, which can be mathematically expressed as follows: for any function $\beta(\bs)$ defined on $\mR$ to be a varying coefficient function within model \eqref{model}, $\beta(\bs)$ must satisfy that:
	\begin{itemize}
		\item[(C1)] there exists an index set $I_1\subset\{1,...,G\}$, such that $\beta(\bs)\times I[\bs\in\overline\mR_g]\in \mC^\rho(\overline\mR_g)$ where $\overline\mR_g$ is the closure of $R_g$, for all $g\in I_1$ with $\rho = \left[\frac{d}{2}\right]+1$;
		\item[(C2)] for any $g\in I_1$, $\beta(\bs)$ is bounded away from zero, that is, $$\lambda=\inf_{\bs\in \cup_{g\in I_1}\mR_g}|\beta(\bs)|>0;$$
		\item[(C3)] let $I_0=\{1,...,G\}\backslash I_1$, then $\beta(\bs)=0$ for all $\bs\in \cup_{g\in I_0}\mR_g$.
	\end{itemize}
The conditions on the SVCFs can be interpreted as follows: (C1) demonstrates that the functions are smooth within each brain region, which implies more homogeneous covariate effects; (C2) indicates the jump discontinuities at the boundaries of brain regions; (C3) introduces sparsity into each SVCF in model \eqref{model} and restricts the sparsity structure at the regional level.

We introduce the notation $\mP$ to represent a set of functions satisfying (C1)--(C3) defined on $\mR$. The definition $\mP$  carries the information of $R_1, \ldots, R_G$ as well as $\lambda$ and $I_1$. For any function $\beta(\bs)\in \mP$, denote by $I_1\{\beta(\bs)\}$ the index of nonzero sets as given in (C1) and by $\Lambda\{\beta(\bs)\}$ the value $\lambda$ defined in (C2). In the same vein, define a set of $p$-dimensional vector-valued functions $\textbf{P}=\{\bbeta(\bs)=[\beta_1(\bs),\ldots,\beta_p(\bs)]^\top: \beta_k(\bs)\in \mP, k=1,\ldots,p\}$. 

\subsection{The thresholded multiscale Gaussian process priors}
\subsubsection{Construction of the prior}
A Gaussian process (GP) can be regarded as a probabilistic measure on certain functional spaces, making it as popular prior models in Bayesian nonparametric data analysis. In general, the GP prior, denoted by $\mathcal{GP}[\mu(\cdot), C(\cdot, \cdot)]$, is determined by its mean function $\mu: \mR\mapsto \mathbb R$ and the covariance kernel function $C: \mR\times \mR \mapsto \mathbb R$. A draw $\beta(\cdot)\sim \mathcal{GP}[\mu(\cdot), C(\cdot, \cdot)]$ is a function defined on $\mR$ such that any finite collection of its function values are jointly multivariate Gaussian. To be specific, for any choices of $\bs_1,\ldots,\bs_n\in \mR$, $[\beta(\bs_1),\ldots,\beta(\bs_n)]^\top\sim N(\bm\mu, \bm C)$ with $\bm \mu = [\mu(\bs_1),\ldots,\mu(\bs_n)]^\top$ and $\bm C=\{C(\bs_i, \bs_j)\}_{1\leq i\leq n, 1\leq j\leq n}$. The boundedness and smoothness of random functions generated from GP are determined through the covariance kernel function. Typical choices for the covariance kernel functions include but not limited to the rational quadratic kernel, Mat\'{e}rn class of kernels, the square exponential kernel. \citet{stein1999interpolation, williams2006gaussian} contain more detailed discussions on the covariance functions. 

To enable detailed feature selections within the SVCFs $\beta_k(\bs)\in \mP \ (k=1,\ldots, p)$ in model \eqref{model}, we develop the thresholded multiscale Gaussian process (TMGP) prior, which can be represented as follows: $\beta(\bs)\sim \TMGP[\tau^2, \theta^2, \lambda, \kappa(\cdot, \cdot)]$ implies that 
\bg \label{thres}\beta(\bs) = \tbeta(\bs)I_\lambda[\tbeta(\bs)],\ed
\bg \label{mGP}\tbeta(\bs) = \gamma(\bs) + \epsilon(\bs),\ed
\bg \label{gGP}\gamma(\bs)\sim\GP[0, \tau^2\kappa(\bs, \bs')],\ed
\bg \label{lGP}\epsilon(\bs)\sim \GP[0, \theta^2\chi(\bs, \bs')] \ed
for all $\bs\in\mR$, where 
\bg \label{thresfunc}I_\lambda[\tbeta(\bs)]=\sum_{g=1}^G I\left[\bs\in \mR_g: \inf_{s\in \mR_g}|\tbeta(\bs)|>\lambda\right]\ed
is the thresholding function that generate region-wise sparse features; $\lambda>0$ is the thresholding parameter; $\tau^2>0$ and $ \theta^2>0$ are variance parameters in the GPs; $\kappa(\cdot, \cdot): \mR\times \mR\mapsto \mathbb R$ is a kernel correlation function and $\chi(\bs, \bs')$ is constructed from $\kappa$ as $\chi(\bs, \bs')=\sum_{g=1}^G\kappa(\bs, \bs')\times I[\bs, \bs'\in R_g]$. 

The thresholding construction of $\beta(\bs)$ from a particular multiscale Gaussian process~\citep[mGP]{dunson2012multiresolution} $\tbeta(\bs)$ introduces sparsity and enables feature selection. For voxel-wise analysis without regional information, the thresholding function in \eqref{thresfunc} can be simplified as $I_\lambda[\tbeta(\bs)]=I[|\tbeta(\bs)|>\lambda]$, where $\bs$ denotes the center of a voxel. The mGP $\tbeta(\bs)$ in our prior is a combination of one ``global" GP, $\gamma(\bs)$, which captures the general dependence structures across the whole brain domain and one ``local" GP, $\epsilon(\bs)$, reflecting the dependence and variabilities within each parcellated brain region. This multiscale construction can also naturally generate jumping discontinuities on the boundaries of the brain regions. 

We illustrate the procedure to sample an SVCF from our TMGP priors on a two-dimensional square region in Figure \ref{tmgp_eg}, where the dashed lines partition the whole region into four equally spaced sub-regions. $\gamma(\bs)$ is smooth over the whole region; $\epsilon(\bs)$ is smooth within each sub-region but has distinct jumps on the boundaries. As we may see from Figure \ref{tmgp_eg}, one particular issue is that the generating processes, $\gamma(\bs)$ and $\epsilon(\bs)$, could vary substantially while results in similar sparse SVCFs. This phenomenon indicates weak identifiability in global and local GPs in the TMGP prior model. To resolve this issue, we impose the certain constraint on the marginal variance of the local GP $\epsilon(\bs)$. i.e., $\theta^2$. This implies that we assume that the sampled SVCF is constructed from strong global signals and weak local signals. 

\begin{figure}[!t]
	\small
	\makebox[\textwidth][c]{\includegraphics[scale=0.43]{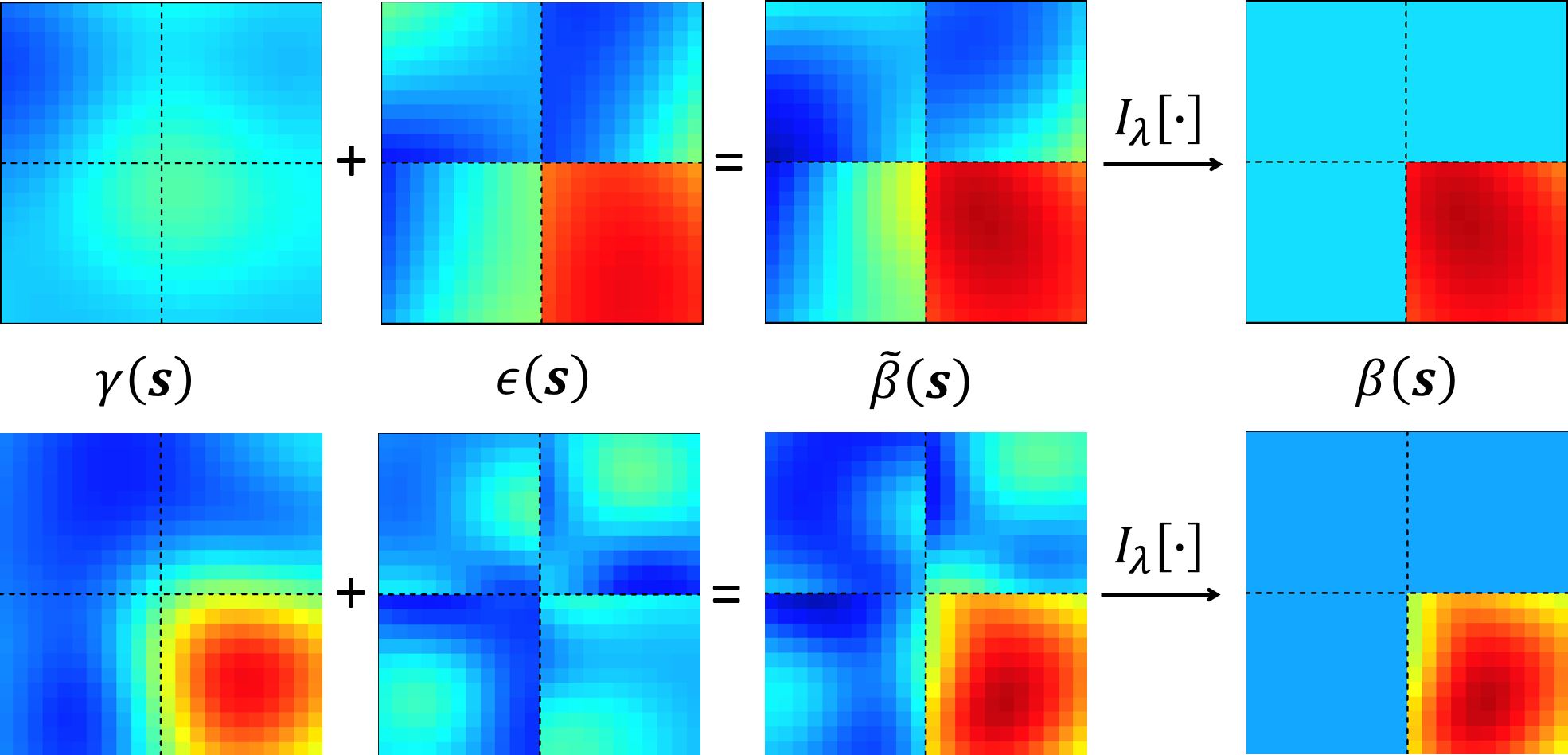}}%
	\caption{Sample SVCFs from TMGP prior. Top row: weak global signals with strong local signals; Bottom row: strong global signals with weak local signals. }\label{tmgp_eg}
\end{figure}

\section{Theoretical results}
We first introduce two sets of extra conditions in addition to the conditions (C1)-(C3) for the SVCFs. 
The design matrix $\bm X$ as defined in \eqref{outcome} satisfies:
\begin{itemize}
	\item[(X1)] Let $d_{\min}$ and $d_\max$ be the smallest and largest eigenvalues of $\frac{1}{m}\bX^\top\bX$, then $0<d_\min<d_\max<\infty$.
\end{itemize}

For the kernel correlation functions $\kappa(\cdot,\cdot)$ in our proposed TMGP priors, we introduce the following condition:
\begin{itemize}
	\item[(K1)] $\kappa(\bs, \bs')=\prod_{j=1}^d K_j(\|\bs-\bs'\|)$ for some nowhere zero, continuous, symmetric density function (up to a normalization constant) $K_j$ defined on $\mathbb R$.
	\item[(K2)] $\kappa(\bs, \cdot)$ has continuous partial derivates up to order $2\rho + 2$ where $\rho=\left[\frac{d}{2}\right]+1$.
\end{itemize}

\begin{theorem}\label{largsup}
	Consider an arbitrary SVCF $\beta^0(\bs)\in\mP$, i.e. $\beta^0(\bs)$ satisfies conditions (C1)-(C3). Let $\lambda_0=\Lambda[\beta^0(\bs)]$ and suppose the partition number $G<\infty$. If $0<\lambda<\lambda_0, 0<\tau^2, \theta^2<\infty$ and the kernel function $\kappa$ satisfies (K1), then the proposed prior $$\beta(\bs)\sim\TMGP[\tau^2, \theta^2, \lambda, \kappa(\cdot, \cdot)]$$ given in \eqref{thres}-\eqref{lGP} satisfies that
	$$\Pi\left(\|\beta(\bs)-\beta^0(\bs)\|_\infty<\varepsilon\right)>0 \quad\quad \mbox{for all  } \varepsilon>0$$
\end{theorem}

Theorem \ref{largsup} demonstrates that the proposed TMGP prior assign positive measures to arbitrarily small neighborhoods of all elements within $\mP$, the family of spatially varying coefficient functions defined in our model \eqref{model}. This property is essential, especially for Bayesian nonparametric  priors, since it is necessary for appropriate posterior behaviors and can not be guaranteed in many cases. 

\begin{theorem}\label{poscons}
	Suppose that our observed data satisfies $\by(\bs_i)\sim N(\bX\bbeta^0(\bs_i), \sigma^2\bm I_m)$ independently following the notations in \eqref{outcome}, with a known $\sigma^2$ and a fixed $p$-dimensional SVCF $\bbeta^0(\bs)=[\beta^0_1(\bs),\ldots,\beta^0_p(\bs)]^{\rT}, \ p<m$, defined on $\mR$. Suppose $\beta^0_k(\bs)$ satisfies conditions (C1)-(C3), i.e. $\bbeta^0(\bs)\in\emph{\textbf P}$, with $G<\infty$; $\bX$ satisfies condition (X1); the kernel function in the TMGP priors satisfies (K1) and (K2).
	For all $\varepsilon>0$, if $\frac{\sigma^2}{m}<\frac{\varepsilon^2 d_{\min}}{8\log2}$ and each dimension of $\bbeta(\bs)$ follows a TMGP prior independently satisfying the conditions in Theorem \ref{largsup}, then the posterior distribution satisfies that
	$$ \Pi\left[U_\varepsilon^c \mid \by(\bs_1),\ldots ,\by(\bs_n)\right] \to 0, $$
	as $n\to\infty$ in $P_{\bbeta^0}^n$ probability, where $U_{\varepsilon}=\left\{\bbeta(\bs)\in  \emph{\textbf P}: \|\bbeta(\bs)-\bbeta^0(\bs)\|_{1, \infty}<\varepsilon\right\}$.
\end{theorem}

Theorem \ref{poscons} justifies the posterior consistency of the proposed TMGP prior given model \eqref{model} under the infill asymptotic framework. It implies that, if a ground truth of the SVCFs exists and the data is generated accordingly, then the posterior distribution of $\bbeta(\bs)$ can be concentrated to an arbitrarily small $\|\cdot\|_{1,\infty}$ neighborhood around the truth as the number of spatial locations goes to infinity. The conditions of this theory also imply that a small ratio between the number of subjects and the variance of pure noise, i.e., $\sigma^2/m$, is also important to guarantee a good performance of our method. One limitation of Theorem \ref{poscons} is that it does not apply to the voxel level analysis where $G=n\to\infty$. However, this type of analysis generally works well empirically. 

Although the $\|\cdot\|_{1,\infty}$ norm is not common in Bayesian asymptotic literatures, a direct result based on Theorem \ref{poscons} is the element-wise posterior consistency for $\bbeta(\bs)$ under the commonly used empirical $\|\cdot\|_1$ norm for a fixed design of spatial locations $\bs_i$ \citep{ghosal2006posterior}. 

\begin{corollary}\label{posconsele}
	Under the same assumptions and conditions in Theorem \ref{poscons}, for all $\varepsilon>0$, if $\sigma^2/m<\frac{\varepsilon^2 d_{\min}}{8\log2}$, then the posterior distribution satisfies that for all $k=1,\ldots,p$,
	$$ \Pi\left[U_{\varepsilon, k}^c \mid \by(\bs_1),\ldots ,\by(\bs_n)\right] \to 0 $$
	as $n\to\infty$ in $P_{\bbeta^0}^n$ probability, where $U_{\varepsilon, k}=\left\{\beta(\bs)\in  \mP: \|\beta(\bs)-\beta_k^0(\bs)\|_{1}<\varepsilon\right\}$.
\end{corollary}

Of note, in neuroimaging studies, $\bs_i$ are usually fixed 3D grid points, thus we do not consider the $\|\cdot\|_1$ norms with regard to random measures for $\bs$ when proving posterior consistency. For a fixed design within a finite domain $\mR$ (the volume of brain is limited), another useful direction is to show posterior consistency under the $\|\cdot\|_1$ norm with regard to the Lebesgue measure. We see this as a potential extension to the theory development in the future work.

\section{Posterior Inferences}
\subsection{Model Representation}
Now consider the SVCM defined in \eqref{model}. For the $p$-dimensional multivariate spatially varying coefficient function, $\bbeta(\bs)=[\beta_1(\bs),...,\beta_k(\bs)]^{\rT}$, we assume that 
$$\beta_k(\bs)\sim\TMGP[\tau_k^2, \theta^2, \lambda_k, \kappa(\cdot,\cdot)],$$
with $\kappa(\cdot,\cdot)$ being a smooth kernel function. 
This specification implies that the global processes \eqref{gGP} have distinct flexible variance parameters $\tau_k^2$, while the local fluctuation processes \eqref{lGP} have a small fixed marginal variance $\theta^2$.  

Based on the prior specification for $\bbeta(\bs)$, we have that  $\beta_k(\bs) = \tbeta_k(\bs)I_{\lambda_k}[\tbeta_k(\bs)]$  for $k=1,\ldots, p$, where
\bg\label{priormodel} \tbeta_k(\bs) = \gamma_k(\bs) +\epsilon_k(\bs),\ed
where $\gamma_k(\bs)\sim\GP[0, \tau_k^2\kappa(\bs, \bs')]$ and $\epsilon_k(\bs)\sim\GP\left(0, \sum_{g=1}^G\theta^2\kappa(\bs, \bs')\times I[\bs, \bs'\in R_g]\right)$. For global GPs: $\gamma_k(\bs)$ in \eqref{priormodel}, its Karhunen-Lo\`{e}ve (KL) expansion can be expressed as 
\bg\label{KLexpan}\gamma_k(\bs) = \sum_{l=1}^\infty \varphi_{l}(\bs)u_{kl},\ed
where $u_{kl} \sim N(0, \tau^2_k\zeta_l)$ independently with $\zeta_l>0$ such that $\sum_{l=1}^\infty \zeta_l \varphi_{l}(\bs) \varphi_{l}(\bs')=\kappa(\bs, \bs')$ and that $ \int\varphi_l(\bs)\varphi_{l^\prime}(\bs)d\bs=0, \forall l\neq l^\prime$  based on the Mercer's Theorem \citep{williams2006gaussian}. In practice, we truncate the infinite sum in \eqref{KLexpan} into $L$ terms such that $\sum_{l=1}^L\zeta_l/\sum_{l=1}^\infty\zeta_l$ is close to $1$. 

This implies a model representation of the proposed SVCM along with the TMGP prior specifications. To be more specific, the  neuroimaging signal $y_j(\bs)$ on locations $\bs_1,\ldots, \bs_n$ can be modeled through latent mGPs: $\bm{\tbeta}(\bs)=[\tbeta_1(\bs),\ldots,\tbeta_p(\bs)]^{\rT}$ and the truncated KL expansion coefficients $\{\bm u_k\}_{k=1}^p$ with $\bm u_k=[u_{k1},...,u_{kL}]^\top$ by integrating out the local GPs $\epsilon_k(\bs)$ in \eqref{priormodel}, which is given by 
\bg\label{lvl1} [y_j(\bs_i)\mid \bm{\tbeta}(\bs_i), \sigma^2] \sim N\left( \bm x_j^\rT g_{\bm\lambda}[\bm{\tbeta}(\bs_i)], \  \sigma^2\right),\ed
\bg \label{lvl2}[\{\tbeta_k(\bs_i)\}_{\bs_i\in \mR_g}\mid \bm u_k] \sim N\left( \bm \varphi_g\bm u_k   ,\  \theta^2 K_g\right),  \ed
\bg \label{lvl3} \ u_{kl}\sim N(0, \zeta_l\tau^2_k), \ed
for $j=1,\ldots,m,\ i=1,\ldots,n, \ k=1,\ldots,p,\  g=1,\ldots,G$ and $l=1,\ldots, L$, where $N(\mu,\sigma^2)$ represents a normal distribution with mean $\mu$ and variance $\sigma^2$, $\bm\varphi_g = \{\bm\varphi(\bs_i)^{\rT}\}_{\bs_i\in\mR_g}$ with $\bm\varphi(\bs_i) =[\varphi_1(\bs_i),\ldots,\varphi_L(\bs_i)]^{\rT}$ and $K_g=\{\kappa(\bs_i, \bs_{i'})\}_{\bs_i,\bs_{i'}\in \mR_g}$ is a correlation matrix. The $p$-dimensional vector value functional operator $g_{\bm\lambda}(\cdot)$ is defined on the domain of all functions in $\mathcal P$, which is given by
  $$g_{\bm\lambda}[\bm{\tbeta}(\bs)]=\left[ \tbeta_1(\bs)I_{\lambda_1}[\tbeta_1(\bs)],\ldots, \tbeta_p(\bs)I_{\lambda_p}[\tbeta_p(\bs)]  \right]^{\rT},$$ 
with  $\bm\lambda=[\lambda_1,\ldots,\lambda_p]^{\rT}$.

\subsection{Hyper Prior Specifications}\label{hyperP}
We assign conjugate priors to variance parameters $\sigma^2$ and $\tau_k^2$ in the TMGP model, i.e. $\sigma^2 \sim \mbox{Inv-Ga}(v, w)$ and $\tau_k^2\iid \mbox{Inv-Ga}(v, w)$, where $\mbox{Inv-Ga}(v, w)$ represents an inverse gamma distribution with shape $v$ and rate $w$.  We fix $\theta^2=1$ to restrict local deviations in order to sort out the weakly identifiability in the model, especially for the case that a small number of observations are recorded in each region.

We develop a data-driven method to specify the prior of thresholding parameters $\lambda_k, k=1,\ldots, p$. We consider the log full conditional of $\lambda_k$ which is given by
$$ \log \pi[\lambda_k\mid {\mathcal Y},   {\bm \tbeta}_k, {\bm \beta}_{-k}, \sigma^2] = \ell[\lambda_k;{\mathcal Y}, {\bm\tbeta}_k, {\bm\beta}_{-k}]/\sigma^2 + C,$$ 
where $C$ is a constant, $\bm\tbeta_k = [\tbeta_k(\bs_1),\ldots, \tbeta_k(\bs_n)]^{\rT}$, $\bm\beta_k = [\beta_k(\bs_1),\ldots, \beta_k(\bs_n)]^{\rT}$, $\bm\beta_{-k} = [{\bm\beta}_{1},\ldots, {\bm\beta}_{k-1}, {\bm\beta}_{k+1},\ldots,{\bm\beta}_{p}]^{\rT}$, and
\bg\label{logliklmd}\ell(\lambda_k) := \ell[\lambda_k;{\mathcal Y}, {\bm\tbeta}_k, {\bm\beta_{-k}}]=\sum_{i=1}^n \omega_k(\bs_i)I_{\lambda_k}[\tbeta_k(\bs_i)], \ed
with $\omega_k(\bs)=\sum_{j=1}^m \tbeta_k(\bs)x_{jk}\left[2y_{j, -k}(\bs)-\tbeta_k(\bs)x_{jk}\right]$ and $y_{j, -k}(\bs) = y_j(\bs)-\sum_{j'\neq k }x_{jj'}\beta_{j'}(\bs)$. The function $\ell(\lambda_k)$  is flat when $\lambda_k$ is around zero and dramatically decreases when $\lambda_k$ is greater a certain value, to which we refer as a ``turning point". It should be close to the true threshold.   Figure \ref{figlmd } shows the profiles of $\ell(\lambda_k)$ for a model with three SVCFs on a space of $900$ locations from $50$ simulated datasets. The true thresholds $\lambda_k = k+1$ for $k = 1,2,3$.  The turning points in the profiles of $\ell(\lambda_k)$ are all around the true thresholds. Thus, we can specify the priors of $\lambda_k$ according to $\ell(\lambda_k)$. In practice, we need to provide rough estimates of $\bm\tbeta_k$ and $\bm\beta_{-k}$ in order to evaluate $\ell(\lambda_k)$ before posterior inferences.  We consider an SVCM with smoothed SVCFs approximated by the truncated K-L expansion, where we compute the ordinary least squares (OLS) of the coefficients, i.e. 
\begin{eqnarray*}
{\{\widehat w_{lk}\}_{l=1}^L}_{k=1}^p = \arg\min_{\{w_{lk}\}}\sum_{j=1}^m \sum_{i=1}^n \left(y_j(\bs_i)  - \sum_{k=1}^p \sum_{l=1}^L  x_{jk} \varphi_l(\bs_i) w_{lk}\right)^2. 
\end{eqnarray*}
Then both $\tbeta_k(\bs)$ and $\beta_k(\bs)$ can be approximated by $\widehat\beta_k(\bs) = \sum_{l=1}^L \varphi_l(\bs) \widehat w_{lk}$. Thus,  we replace $\bm\beta_k$ and $\bm\tbeta_{-k}$  by $\widehat {\bm\beta}_{k} = [\widehat\beta_k(\bs_1),\ldots, \widehat\beta_k(\bs_n)]^{\rT}$ and $\widehat{\bm\beta}_{-k} = [\widehat{\bm\beta}_{1},\ldots, \widehat{\bm\beta}_{k-1}, \widehat{\bm\beta}_{k+1},\ldots,\widehat{\bm\beta}_{p}]^{\rT}$, respectively, in \eqref{logliklmd}.  Write $\widehat\ell(\lambda_k) = \ell(\lambda_k; \mathcal Y, \widehat{\bm\beta}_k, \widehat{\bm\beta}_{-k})$.

We propose to assign uniform priors to $\lambda_k$, i.e. $\lambda_k \sim \mathrm{Unif}(c_k-h_k, c_k+h_k)$, where the half range $h_k$  and center $c_k$ can be determined based on the profile of $\widehat\ell(\lambda_k)$. More specifically,  we evaluate $\widehat\ell(\lambda_k)$ on a set of grid points $\{\lambda^{(1)}_k,\ldots,\lambda^{(G)}_k\}$, denoted $\{\ell^{(1)}_k, \ldots, \ell^{(G)}_k\}$. Given an interval $(a, b)$, define the sample correlation between $\lambda_k$ and $\widehat\ell(\lambda_k)$ within $(a,b)$ as
\begin{eqnarray*}
\widehat\rho(a,b) = \frac{\sum_{\lambda_k^{(g)}\in (a,b)}(\lambda^{(g)}_k -\overline\lambda_k)(\ell^{(g)}_k-\overline\ell_k)}{\sqrt{\sum_{\lambda_k^{(g)}\in (a,b)}(\lambda^{(g)}_k -\overline\lambda_k)^2}\sqrt{\sum_{\lambda_k^{(g)}\in (a,b)}(\ell^{(g)}_k -\overline\ell_k)^2} }
\end{eqnarray*}
with $\overline\lambda_k = \sum_{\lambda_k^{(g)}\in (a,b)} \lambda^{(g)}_k/ M(a,b)$, $\overline\ell_k = \sum_{\lambda_k^{(g)}\in (a,b)} \ell^{(g)}_k/ M(a,b)$ and $M(a,b) =  \sum_{g=1}^GI[\lambda_k^{(g)}\in (a,b)]$. And define 
\begin{eqnarray*}
\widetilde c_k(h) = \min\{\lambda^{(g)}_k: |\widehat\rho(\lambda_k^{(g)}-h, \lambda_k^{(g)}+h)|> \zeta_k\}, 
\end{eqnarray*}
where $\zeta_k$ is determined by the rejection region of Pearson correlation test. Given $h>0$, $\widetilde c_k(h)$ represents the location that $\widehat\ell(\lambda_k)$ and $\lambda_k$ have no significant correlation. Then we specify 
\begin{eqnarray*}
h_k = \min\{h: |\widehat\rho(\widetilde c_k(h)-h, \widetilde c_k(h)+h)|> \zeta_k \},\quad \mbox{and} \quad c_k = \widetilde c_k(h_k).
\end{eqnarray*}
This leads to an informative prior range $[h_k-c_k, h_k+c_k]$ for $\lambda_k$ with a high probability to cover the turning point of $\ell(\lambda_k)$.

  

%
%

\begin{figure}[!t]
	\small
	\makebox[\textwidth][c]{\includegraphics[scale=0.38,  angle =270]{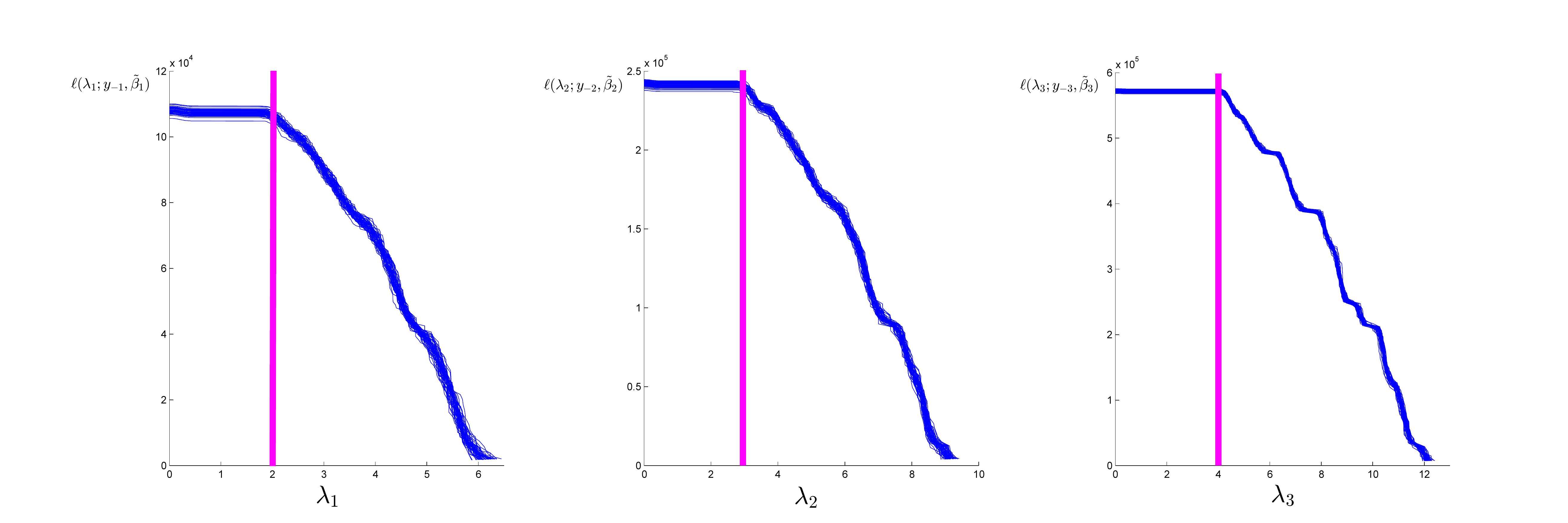}}%
	\caption{Simulated $\widehat\ell(\lambda_k)$ from $50$ synthetic datasets: ground truth ($\lambda_1=2, \lambda_2=3, \lambda_3=4$) are marked in the figures. }\label{figlmd }
\end{figure}



\subsection{Kernel Expansion for Massive Data Analysis}
Theoretically, the KL expansion for the GP with kernel function $\kappa(\cdot, \cdot)$ relies on solving the integral equation $\int \kappa(\bs, \bs')\varphi(\bs)d\bs = \zeta_l\varphi_(\bs')$, which might not admit analytical solutions. Empirically, the expansion is often achieved by calculating the eigenvalues and eigenvectors of the $n\times n$ correlation matrix $K_n=\{\kappa(\bs_i, \bs_{i'})\}_{1\leq i, i'\leq n}$ on a set of pre-specified locations. However, in the analysis of massive neuroimaging data that can involve a very large number ($n$ can be hundreds of thousands) of brain locations, it is computationally infeasible to perform eigen decompositions on $K_n$. To solve this issue, we introduce the modified square exponential kernel
\bg \label{mse} \kappa(\bs, \bs') = \exp\{-a\|\bs\|_2^2-a\|\bs'\|_2^2-b\|\bs-\bs'\|_2^2\}, \ a, b>0 \ed
with a relatively small value for $a$ as a numerical approximation to the square exponential kernel when dealing with massive neuroimaging data. The major benefit of this kernel function is that it has analytically tractable expansion. The detailed properties of this kernel is summarized in Proposition \ref{kernel}.
\begin{proposition}
For a specific $l\in\{1,...,\infty\}$, define series $\{k_i\}_{i=0}^{d}$, $\{l_i\}_{i=0}^d$ and $\{m_i\}_{i=1}^d$ as follows
$$k_i=\left\{k_i\in \mathbb N^0: {k_i+d-i-1\choose d-i}\leq l_i \leq {k_i+d-i\choose d-i}-1\right\},\  0\leq i \leq d-1,\ k_d =0,$$
$$l_0=l-1, \ l_i = l_{i-1}-{k_{i-1}+d-i\choose d-i+1}, \ i\geq 1 ,$$
$$ \ m_i = k_{i-1}- k_i, \ i\geq 1, $$
where $\mathbb N^0$ is the set of nonnegative integers; ${n\choose k}=0$ if $k>n$. Define $L ={m+d\choose d}=\sum_{k=0}^m {k+d-1\choose d-1} $. For $\bs=[s_1,...,s_d]^\top\in\mathbb R^d$, let $\varphi_l(\bs)$ and $\zeta_l\ ( l=1,...,\infty)$ be the eigenfunctions and eigenvalues for the modified square exponential kernel $\kappa(\bs, \bs')$ as defined in \eqref{mse}, then 
$$\zeta_l = \left(\frac{\pi}{A}\right)^d B^{k_0}, \quad \frac{\sum_{l=1}^L\zeta_l}{\sum_{l=1}^\infty\zeta_l}= (1-B)^d\sum_{k=0}^m {k+d-1\choose d-1}B^k, $$
$$ \varphi_l(\bs) =(2c)^{\frac{d}{4}}\exp\left(-c\|\bs\|_2^2\right)\prod_{i=1}^d H_{m_i}(\sqrt{2c}s_i), $$
where $c = \sqrt{a^2+2ab}$, $A= a + b + c$ and	$B = b/A$; $H_k(\cdot)$ is the $k$th ($k\in \mathbb N^0$) order normalized hermit polynomial, which is defined by 
$ H_k(x) = (2^k k! \sqrt{\pi})^{-1/2} (-1)^k \exp(x^2) \frac{d^k}{d x^k} \exp(-x^2)$. 
\label{kernel}
\end{proposition}

\subsection{A Markov chain Monte Carlo Algorithm}
We developed a generally efficient MCMC sampling algorithm for posterior inference about $[\{\bm{\tbeta}(\bs_i)\}_{i=1}^n,  \ \{\bm u_k\}_{k=1}^p, \ \sigma^2, \ \{\tau_k^2\}_{k=1}^p, \  \bm \lambda\mid \mathcal Y]$ based on the representation and approximation for our model with the TMGP priors \eqref{lvl1}-\eqref{lvl3}. 

Updating $\bm{\tbeta}(\bs_i), \ i=1,\ldots,n,$ is an essential step in the MCMC algorithm. The Metropolis-Hasting (M-H) algorithm is employed with a block updating scheme separately for $\{\tbeta_k(\bs_i)\}_{\bs_i\in\mR_g}$, $g=1,\ldots,G$ to facilitate efficient chain mixing. Under the scenario where region partition structure $\mR_1,\ldots,\mR_G$ is available and reliable, we can directly use this partition information. For voxel level analysis or analysis where no prior knowledge about the regional information are adopted, we first fit voxel-wise GLMs and then use certain clustering algorithms to cluster the resulting spatially varying coefficient values. This initial clustering results for the brain locations (usually centers of voxels) are used for block updating. Another M-H step in our algorithm is updating $\lambda_k, k=1,\ldots,p$, the thresholding parameters in the TMGP priors. The remaining parameters ($\bm u_k, \ k=1,\ldots, p$) and hyperparameters ($\sigma^2, \tau_k^2, \k=1,\ldots, p$) are updated by directly drawing samples from their full conditionals due to conjugacy. More details about the MCMC algorithm are available at the appendix.

\subsection{Posterior Inference on SVCFs}
With the recorded MCMC samples $\tbeta_k^{(t)}(\bs_i)$, $\bm u_k$ and $\lambda_k^{(t)}, \ t=1,\ldots,T$, we can achieve three major goals: 1) selecting neuroimaging features; 2) estimating covariate effects at the feature regions; 3) making prediction on the covariate effects at any brain location.

To select the important imaging features at the regional level, we estimate the selection probability of every region, $g=1,\ldots, G$,  according to the definition (C1)-(C3), using the MCMC samples as 
$$\widehat P(g\in I_1) = \widehat P\left(\inf_{1\leq i\leq n, \bs_i\in\mR_g}  |\tbeta_k(\bs_i)|>\lambda_k\mid\mathcal Y\right)\approx \frac{1}{T}\sum_{t=1}^T I\left[\inf_{1\leq i\leq n, \bs_i\in\mR_g}|\tbeta_k^{(t)}(\bs_i)|>\lambda_k^{(t)}\right], $$
then we estimate $\bbeta(\bs_i)$ as follows, if $\bs_i\in\mR_g$,

\bg\label{estimator}\hat\beta_k(\bs_i)=\left\{
\begin{array}{lc}
	\widehat E[\tbeta_k(\bs_i)\mid  \inf_{1\leq j\leq n, \bs_j\in\mR_g}  |\tbeta_k(\bs_j)|>\lambda_k, \mathcal Y], & \widehat P(g\in I_1)>q\\
	0, & \widehat P(g\in I_1)\leq q
\end{array}
\right., \ed
for all $k=1,...,p$, where $0.5<q<1$ is a threshold for the posterior probabilities of being nonzero at certain brain locations. We use $q=0.90$ throughout the rest of our analysis. Estimates for the posterior conditional expectations in \eqref{estimator} can be easily calculated based on the posterior samples. 

As a special case, to conduct voxel level selection (i.e., each voxel is a region with voxel centers being $\bs_1,\ldots,\bs_n$), we can simply adapt \eqref{estimator} to 
\bg\label{estimatorvxl}\hat\beta_k(\bs_i)=\left\{
\begin{array}{lc}
	\widehat E[\tbeta_k(\bs_i)\mid  |\tbeta_k(\bs_i)|>\lambda_k, \mathcal Y], & \widehat P(|\tbeta_k(\bs_i)|>\lambda_k\mid \mathcal Y)>q\\
	0, & \widehat P(|\tbeta_k(\bs_i)|>\lambda_k\mid \mathcal Y)
	\leq q
\end{array}
\right.,\ed
where $\widehat P(|\tbeta_k(\bs_i)|>\lambda_k\mid \mathcal Y)\approx \frac{1}{T}\sum_{t=1}^T I\left[|\tbeta_k^{(t)}(\bs_i)|>\lambda_k^{(t)}\right] $ can be regarded as the posterior probability of activation for each voxel. 

Making predication on $\bbeta(\cdot)$ at an arbitrary new brain location $\bs_0\in \mR$ with the posterior samples is also available. Without loss of generality, suppose $\bs_0\in \mR_g$, then
$$p(\tbeta_k(\bs_0)\mid \mathcal Y)= \int p(\tbeta_k(\bs_0)\mid \bm \tbeta_k^g, \mathcal Y)\times p( \bm \tbeta_k^g\mid \mathcal Y)d \bm\tbeta_k^g $$
where $p(\cdot)$ represents the probability densities; $\bm \tbeta_k^g$ is actually $\{\tbeta_k(\bs_i)\}_{\bs_i\in \mR_g}$; 
$[\tbeta_k(\bs_0)\mid \bm \tbeta_k^g, \mathcal Y]=[\tbeta_k(\bs_0)\mid \bm \tbeta_k^g]\sim N(\bm\varphi(\bs_0)^\top\bm u_k + k_g(\bs_0)^\top K_g^{-1}(\bm\tbeta_k^g-\bm\varphi_g\bm u_k),\  \theta^2-\theta^2k_g(\bs_0)^\top K_g^{-1}k_g(\bs_0))$ with $k_g(\bs_0)=\{\kappa(\bs_0, \bs_i)\}_{\bs_i\in\mR_g, 1\leq i\leq n}$. This indicates that we can predict $\tbeta_k(\cdot)$ at an unobserved location $\bs_0$ as

\bg\label{predict}\hat\beta_k(\bs_0)=\left\{
\begin{array}{ll}
 \frac{1}{T}\sum_{t=1}^\top [ \bm\varphi(\bs_0)^\top\bm u_k^{(t)} + k_g(\bs_0)^\top K_g^{-1}(\bm\tbeta_k^{g(t)}-\bm\varphi_g\bm u_k^{(t)})]\bm\tbeta_k^{g(t)}  , &  \mR_g \ \mbox{is selected}\\
	0, &  \mR_g \ \mbox{is not selected}
\end{array}
\right.,\ed
where $\bm\tbeta_k^{g(t)}$ is the corresponding vector for $\bm{\tbeta}_k^g$ based on the $t$th recorded MCMC sample.

\section{Numerical Examples}
\subsection{Simulation Study: Synthetic 2D Imaging Data}
For demonstration purpose, we consider a 2D case, i.e. $\mR\subset [0,1]^2$ in this simulation. three covariate functions, $\beta_1(\bs), \beta_2(\bs)$ and $\beta_3(\bs)$, are created on $\mR$ as shown in Figure \ref{fig1} (the ``true" column). We consider $n=30\times 30, \ 50\times 50$ and $60\times60$ locations within the fixed domain $\mR$. The data is generated from
\bg\label{simu1} y_j(\bs_i) = \beta_1(\bs_i)x_{j1}+\beta_2(\bs_i)x_{j2}+ \beta_3(\bs_i)x_{j3}+e_j(\bs_i),\ed
where $e_j(\bs_i)\iid N(0, \sigma^2)$ and $x_{j1}\iid N(0, 4), x_{j2}\iid \mbox{Unif}(-1, 1), x_{j3}\iid \mbox{Bernoulli}(0.5)$. We considered two sample sizes $m=50$ and $m=100$ in combination with two different noise levels $\sigma^2=2$ and $\sigma^2=4$. Given one combination $\{n,\  m,\ \sigma^2\}$, $50$ datasets are independently generated. The MCMC algorithm is implemented to fit the SVCM model and select informative features for each dataset. We choose the SE kernel for the TMGP priors. The spatial range parameter, $b$, was fixed as $30$. The priors for the thresholding parameters were fixed as Unif$(0.3, 1.25)$. All pixels are divided into four subsets for block updating according to $k$-means clustering results. The MCMC iterations are implemented $10,000$ times with the first $5,000$ samples discarded as burn-in. For each simulated dataset, the algorithm usually takes less than $1$ minute to complete the $10,000$ iterations on a standard Intel i7 quad core desktop PC.

We compare our results to the standard voxelwise GLM methods. The features estimated from the GLM method are thresholded based on the $p$-values for testing whether $\beta_k(\bs_i)=0$. We considere both the direct thresholding based on na\"{i}ve $t$-test (GLM-$t$), thresholding using the FDR adjusted $p$-values (GLM-FDR) \citep{benjamini2001control, benjamini2006adaptive} and thresholding by controlling FWER based on standard random field theory (GLM-RFT) \citep{nichols2003controlling}. Figure \ref{fig1} presents the estimated covariate effect functions from GLM-$t$, GLM-FDR, GLM-RFT and our method based on one simulated dataset from our experiments.  Figure \ref{fig1} shows that our method provides more accurate feature selection results by eliminating most false positive signals as well as maintaining high sensitivity. Our estimates also recover the features more accurately by incorporating spatial smoothness. In addition, our posterior inference procedures can provide, for each pixel in the images, the probability of presenting features that are different from zero, as shown in the last column in Figure~\ref{fig1}. For all the scenarios, we report in Table \ref{table1} the relative mean square errors with regard to the GLM estimates, which is defined as
\bg \mbox{ReMSE} = \frac{\sum_{i=1}^n\sum_{k=1}^p\left[\hat\beta_k(\bs_i)-\beta_k(\bs_i)\right]^2}{\sum_{i=1}^n\sum_{k=1}^p\left[\hat\beta_k^{*}(\bs_i)-\beta_k(\bs_i)\right]^2},\ed
where $\hat\beta_k(\bs_i)$ are the estimates from a certain method, $\hat\beta_k^*(\bs_i)$ are the voxel-wise GLM estimates without any thresholding and $\beta_k(\bs_i)$ represent the true values. We also reported the false discovery rates (FDRs) and the false negative rates (FNRs) in Table \ref{table1}, which are specified as:
\bg \mbox{FDR} = \frac{\sum_{i=1}^n\sum_{k=1}^p I[\hat\beta_k(\bs_i)\neq 0]\times I[\beta_k(\bs_i)= 0] }{\sum_{i=1}^n\sum_{k=1}^p I[\hat\beta_k(\bs_i)\neq 0]},\ed
\bg \mbox{FNR} = \frac{\sum_{i=1}^n\sum_{k=1}^p I[\hat\beta_k(\bs_i)= 0]\times I[\beta_k(\bs_i)\neq 0] }{\sum_{i=1}^n\sum_{k=1}^p I[\beta_k(\bs_i)\neq 0]}.\ed

Based on the results from Table \ref{table1}, our method performs well consistently in terms of both feature selection (small FDRs and FNRs) as well as estimation (small ReMSE), especially when the noise level is high or the number of subjects is small. Random field theory based thresholding performs well at low noise level but deteriorates notably as noise level increases due to low sensitivity. FDR control is also relatively robust but this method consistently generates false positive signals. The performance of our SVCM-TMGP method also increases as the number of spatial locations increases within a fixed domain, which agrees with our posterior consistency theory based on infill asymptotics. 

To comprehensively compare the performance of TMGP priors in feature selection with other common methods, we also conduct the receiver operating characteristic (ROC) analysis. Since our original method will automatically generate the optimal thresholding values, in this ROC analysis, we fix $\bm \lambda$ at different values and rerun the MCMC simulation to alternate the specificities. Figure \ref{roc} shows the ROC curves of our method, GLM-FDR and GLM-RFT with $\sigma^2=4 $ and $m=50$ ($n=900, 2500$ and $3600$, respectively). To quantify the differences, we calculate the area under ROC curves (Table \ref{tab:roc}) when the false positive rates are smaller than $0.1$ for all three methods.  Figure \ref{roc} and Table \ref{tab:roc} show that, under all these three settings, our method (SVCM-TMGP) achieve the best performance; as $n$ increases, our method proves to have improved performance. The random field correction based on FWER control (GLM-RFT) suffers from serious false negative problems, leading to low statistical powers. The FDR control (GLM-FDR) is a competitive alternative to our method, especially when the number of spatial locations, $n$, is relatively small. 

To demonstrate the Bayesian learning of the thresholding parameters, we present the histograms for our recorded MCMC samples along with the trace plots for the whole Markov chain from one simulated dataset. It is clear that the marginal posterior distributions of the all thresholding parameters are different from the same prior Unif$(0.3, 1,25)$.  This indicates our model can achieve Bayesian learning \citep{xie2004note} of all the thresholding parameters.

\begin{figure}[!t]
	\small
	\makebox[\textwidth][c]{\includegraphics[scale=0.38,angle=270]{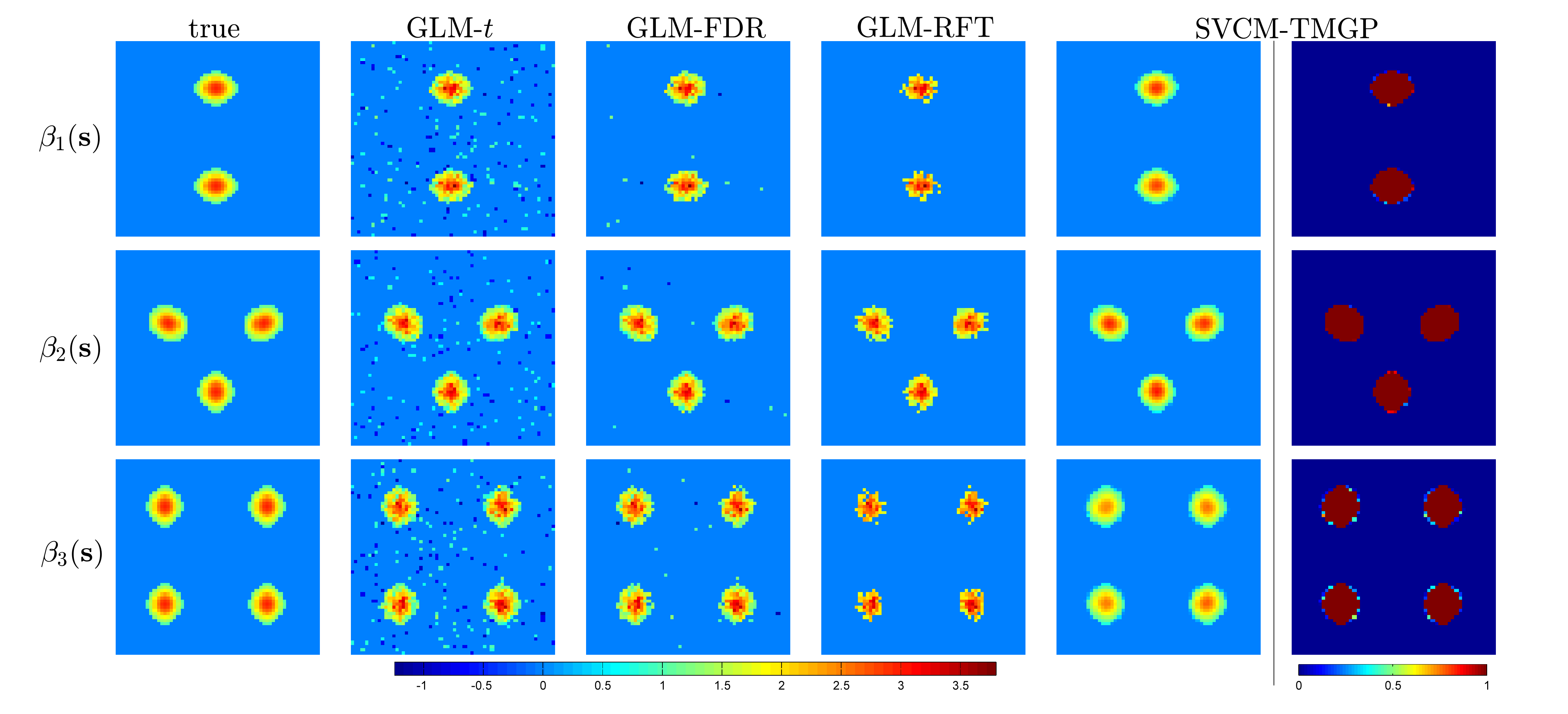}}%
	\caption{Column 1-5: true and estimated spatial covariate effects from GLM-$t$, GLM-FDR, GLM-RFT and SVCM-TMGP; Column 6: the selection probability estimated from SVCM-TMGP. The result is generated from one simulated dataset with $m=50$ subjects, $n=3600$ pixels and noise level $\sigma^2=4$.}\label{fig1}
\end{figure}

\begin{table*}[!t]
	\caption{Quantitative comparison of SVCM-TMGP to voxel-wise GLM fitting results with various thresholdings. All results reported are the means and standard errors based on $50$ independently simulated datasets.}
	\label{table1}
	\vskip 0.15in
	\begin{center}
		\begin{small}
			\footnotesize
			\begin{tabulary}{1.0\textwidth}{ccccccccccccccc}
				\toprule

				          &             && ReMSE&FDR(\%)&FNR(\%)   
				                        && ReMSE&FDR(\%)&FNR(\%)      
				                          \\\midrule
		         
		       n=900&  &&\multicolumn{3}{c}{$(\sigma^2=2,\ m=50)$}       &&\multicolumn{3}{c}{$(\sigma^2 = 4,\ m=50)$} \\
		        \cline{4-6}\cline{8-10}
				
				&GLM-$t$      
				&& 0.39(0.12)   &18.3(4.3)  & 0.0(0.0)
				&& 0.40(0.23)   &17.5(5.3)  & 1.5(0.2)
				\\
				&GLM-FDR      
				&& 0.24(0.07)   & 5.0(1.0)   & 0.0(0.0)
				&& 0.32(0.08)   & 5.4(1.5)   & 4.1(0.6)
				\\
				&GLM-RFT     
				&& 0.36(0.17)   & 0.0(0.0)  & 4.4(1.5)
				&& 0.80(0.29)   & 0.0(0.0)  & 22.7(5.6)
				\\      
				&SVCM-TMGP    
				&& 0.19(0.08)   & 0.9(0.4)  & 0.8(0.2)
				&& 0.12(0.02)   & 2.8(0.9)  & 0.7(0.2)\\

				  n=900&  &&\multicolumn{3}{c}{$(\sigma^2=2,\ m=100)$}         &&\multicolumn{3}{c}{$(\sigma^2 = 4,\ m=100)$} \\
				     \cline{4-6}\cline{8-10}
			   	  &GLM-$t$      
			   	  && 0.41(0.13)   &18.9(3.5)  & 0.0(0.0)   
			  	  && 0.39(0.11)   &19.2(4.9)  & 0.0(0.0)   
			  	    \\
			  	  &GLM-FDR     
			  	  && 0.25(0.06)   & 5.2(0.7)   & 0.0(0.0)       
			  	  && 0.22(0.06)   & 5.1(1.2)   & 0.0(0.0) 
			     	\\
			  	  &GLM-RFT      
			  	  && 0.17(0.10)   & 0.2(0.0)   & 0.0(0.0)       
			  	  && 0.31(0.14)   & 0.5(0.0)   & 4.6(1.3)  
			  	
			  	    \\      
			  	  &SVCM-TMGP   
			  	  && 0.20(0.08)   & 0.2(0.0)   & 0.2(0.0)       
			  	  && 0.19(0.08)   & 1.3(0.4)   & 0.6(0.2) 
			  	\\\midrule
			       
			       n=2500&  &&\multicolumn{3}{c}{$(\sigma^2=2,\ m=50)$}       &&\multicolumn{3}{c}{$(\sigma^2 = 4,\ m=50)$} \\
			            \cline{4-6}\cline{8-10}
			       
			        &GLM-$t$      
			       	&& 0.32(0.10)   &28.1(5.3)  & 0.0(0.0)
			       	&& 0.33(0.16)   &28.8(6.5)  & 0.4(0.1)
			       	\\
			       	&GLM-FDR      
			       	&& 0.13(0.02)   & 2.9(1.0)  & 0.3(0.1)
			       	&& 0.20(0.08)   & 4.3(1.6)  & 5.0(1.6)
			       	\\
			       	&GLM-RFT     
			       	&& 0.27(0.14)   & 0.3(0.0)  & 6.3(2.1)
			       	&& 0.61(0.29)   & 0.4(0.2)  & 27.9(4.5)
			       	\\      
			       	&SVCM-TMGP   
			       	&& 0.12(0.06)   & 0.2(0.0)  & 0.6(0.1)
			       	&& 0.08(0.03)   & 1.2(0.4)  & 0.9(0.4)\\
			       			
			       n=2500&  &&\multicolumn{3}{c}{$(\sigma^2=2,\ m=100)$}       &&\multicolumn{3}{c}{$(\sigma^2 = 4,\ m=100)$} \\
			        \cline{4-6}\cline{8-10}
			       
			       	&GLM-$t$      
			       	&& 0.34(0.09)   &29.4(4.9)  & 0.0(0.0)   
			     	&& 0.35(0.12)   &29.2(5.0)  & 0.0(0.2)
			     	\\
			     	&GLM-FDR      
			     	&& 0.19(0.06)   & 5.4(0.8)   & 0.0(0.0)       
			     	&& 0.14(0.03)   & 4.6(1.3)   & 0.1(0.0)
			     	
			     	\\
			     	&GLM-RFT      
			     	&& 0.10(0.05)   & 0.1(0.0)   & 0.1(0.0)       
			     	&& 0.30(0.14)   & 0.3(0.0)   & 7.0(2.3)
			     	
			     	\\      
			     	&SVCM-TMGP    
			     	&& 0.15(0.04)   & 0.1(0.0)   & 0.3(0.0)       
			     	&& 0.13(0.05)   & 0.4(0.1)   & 0.6(0.1) 
			     	\\\midrule
       
			      n=3600&  &&\multicolumn{3}{c}{$(\sigma^2=2,\ m=50)$}       &&\multicolumn{3}{c}{$(\sigma^2 = 4,\ m=50)$} \\
			           \cline{4-6}\cline{8-10}
			       
			       &GLM-$t$      
			       && 0.33(0.09)   &35.2(6.2)   & 0.0(0.0)
			       && 0.31(0.11)   &31.5(7.3)   & 0.7(0.4)
			       \\
			       &GLM-FDR      
			       && 0.12(0.03)   & 3.8(1.0)   & 0.2(0.0)
			       && 0.17(0.04)   & 3.8(1.3)   & 5.2(1.4)
			       \\
			       &GLM-RFT     
			       && 0.25(0.07)   & 0.2(0.0)   & 7.6(2.2)
			       && 0.48(0.10)   & 0.3(0.1)   & 28.7(4.6)
			       \\      
			       &SVCM-TMGP
			       && 0.10(0.03)   & 0.1(0.0)   & 0.0(0.0)
			       && 0.07(0.02)   & 1.1(0.3)   & 1.8(0.6)\\
			       
			       n=3600&  &&\multicolumn{3}{c}{$(\sigma^2=2,\ m=100)$}       &&\multicolumn{3}{c}{$(\sigma^2 = 4,\ m=100)$} \\
			       \cline{4-6}\cline{8-10}
			       
			       &GLM-$t$      
			       && 0.33(0.06)   &34.4(5.1)   & 0.0(0.0)   
			       && 0.33(0.09)   &33.5(6.0)  & 0.0(0.0)
			       \\
			       &GLM-FDR      
			       && 0.18(0.02)   & 4.3(0.5)   & 0.0(0.0)       
			       && 0.13(0.03)   & 2.6(0.7)   & 0.0(0.0)
			       
			       \\
			       &GLM-RFT      
			       && 0.09(0.06)   & 0.0(0.0)   & 0.2(0.5)       
			       && 0.24(0.07)   & 0.1(0.0)   & 6.5(1.3)
			       
			       \\      
			       &SVCM-TMGP
			       && 0.12(0.03)   & 0.0(0.0)   & 0.5(0.0)       
			       && 0.11(0.03)   & 0.2(0.0)   & 0.5(0.3) \\			    		  				  	
                      \bottomrule
			\end{tabulary}
		\end{small}
	\end{center}
	\vskip -0.1in
\end{table*}

\begin{figure}[!t]
	\begin{center}
		\small$
		\begin{array}{ccc}
		\includegraphics[scale=0.45, angle=270]{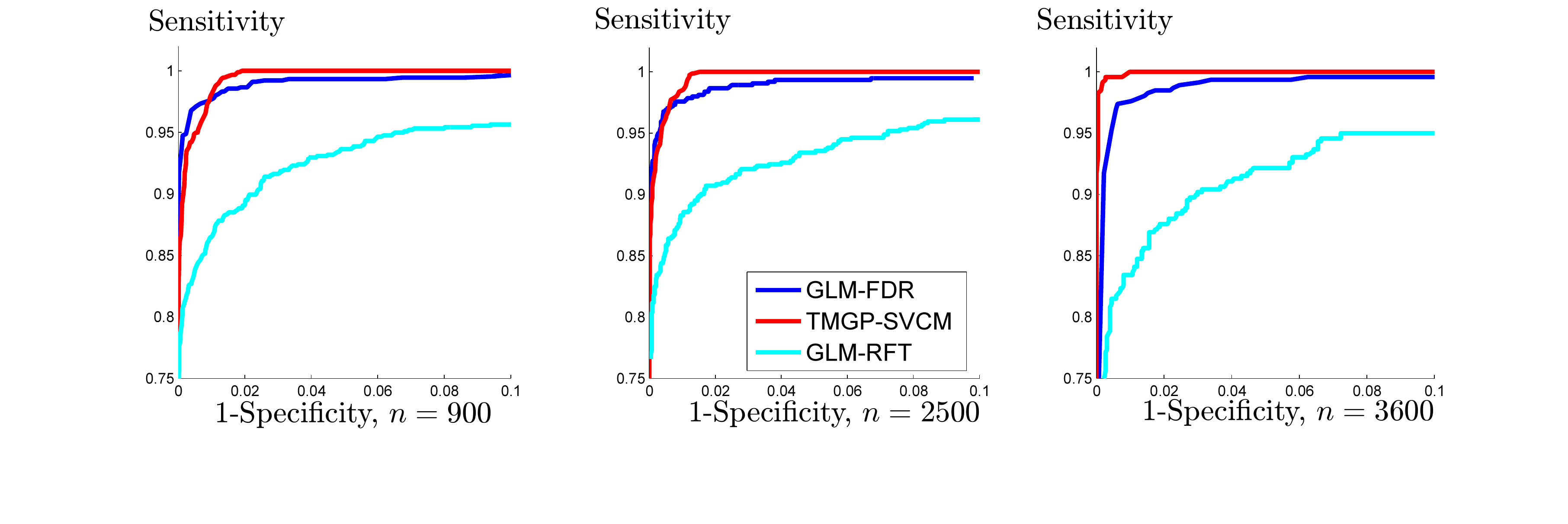}&
		\end{array}$
		\caption{ROC analysis results from $50$ replicated datasets. Three competing methods are TMGP-SVCM, GLM-FDR, GLM-RFT. The variance are all $\sigma^2=4$; subject numbers are all $m=50$; the number of spatial locations are $n=900, 2500$ and $3600$.}\label{roc}
	\end{center}
\end{figure}

\begin{table*}[!t]
\caption{Area under the ROC curves (AUC) with false positive rates are within $[0, 0.1]$}
\label{tab:roc}
\vskip 0.15in
\begin{center}
\begin{small}
\footnotesize
\begin{tabulary}{1.0\textwidth}{lccccccccccc}
\toprule
& \multicolumn{3}{c}{AUC$\times 10^{-1}$}\\
&              $n=900$ & $n=2500$ & $n=3600$
\\
\midrule
GLM-FDR  &0.978 (0.010)&0.976 (0.014)& 0.979 (0.005)\\
GLM-RFT  &0.912 (0.049)&0.915 (0.080)& 0.911 (0.072)\\
SVCM-TMGP&0.982 (0.012)&0.989 (0.007)& 0.998 (0.001)\\
\bottomrule
\end{tabulary}
\end{small}
\end{center}
\vskip -0.1in
\end{table*}

\begin{figure}[!t]
	\begin{center}
		\small$
		\begin{array}{cc}
		\includegraphics[scale=0.38]{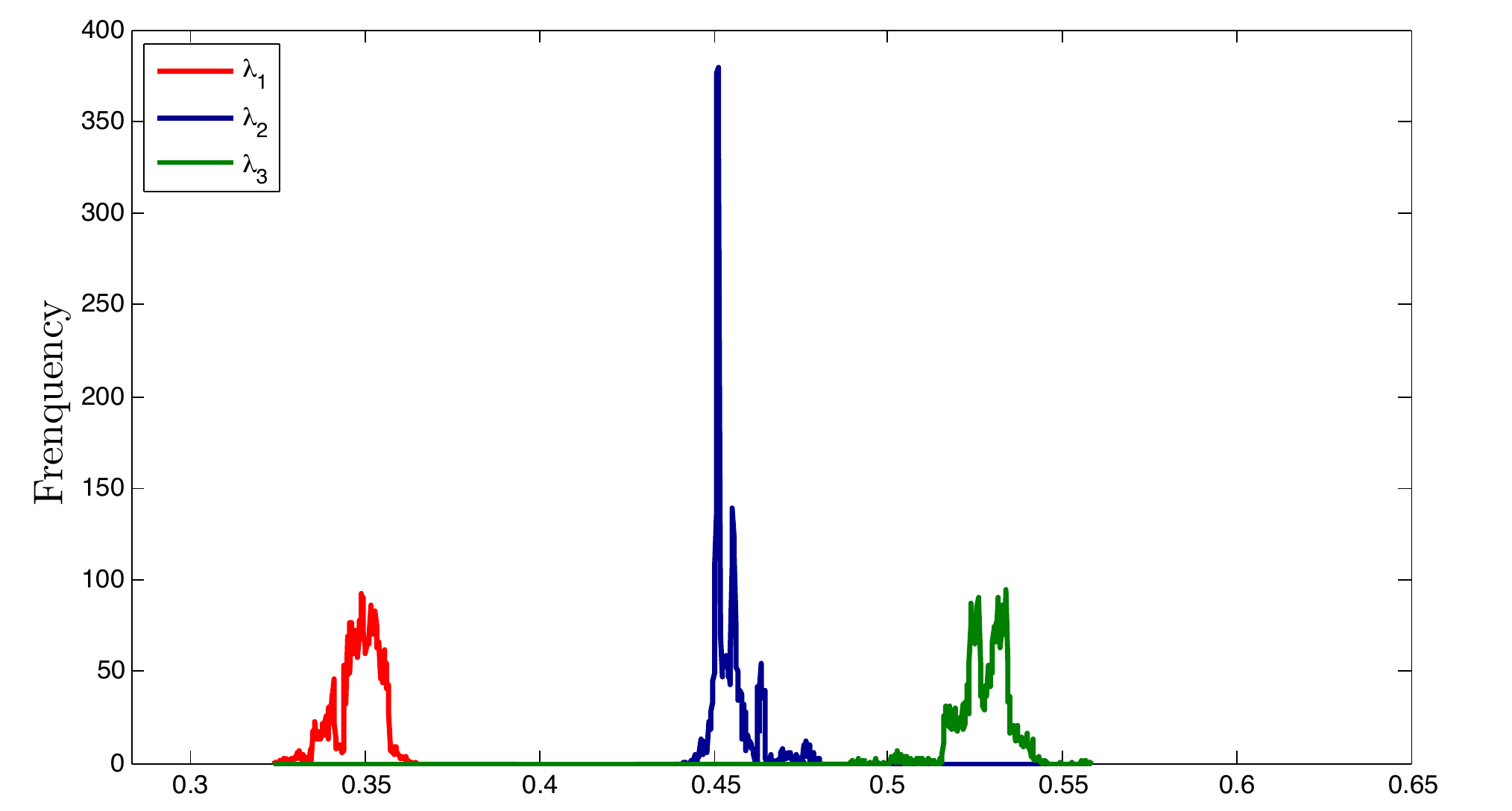}&
		\includegraphics[scale=0.35]{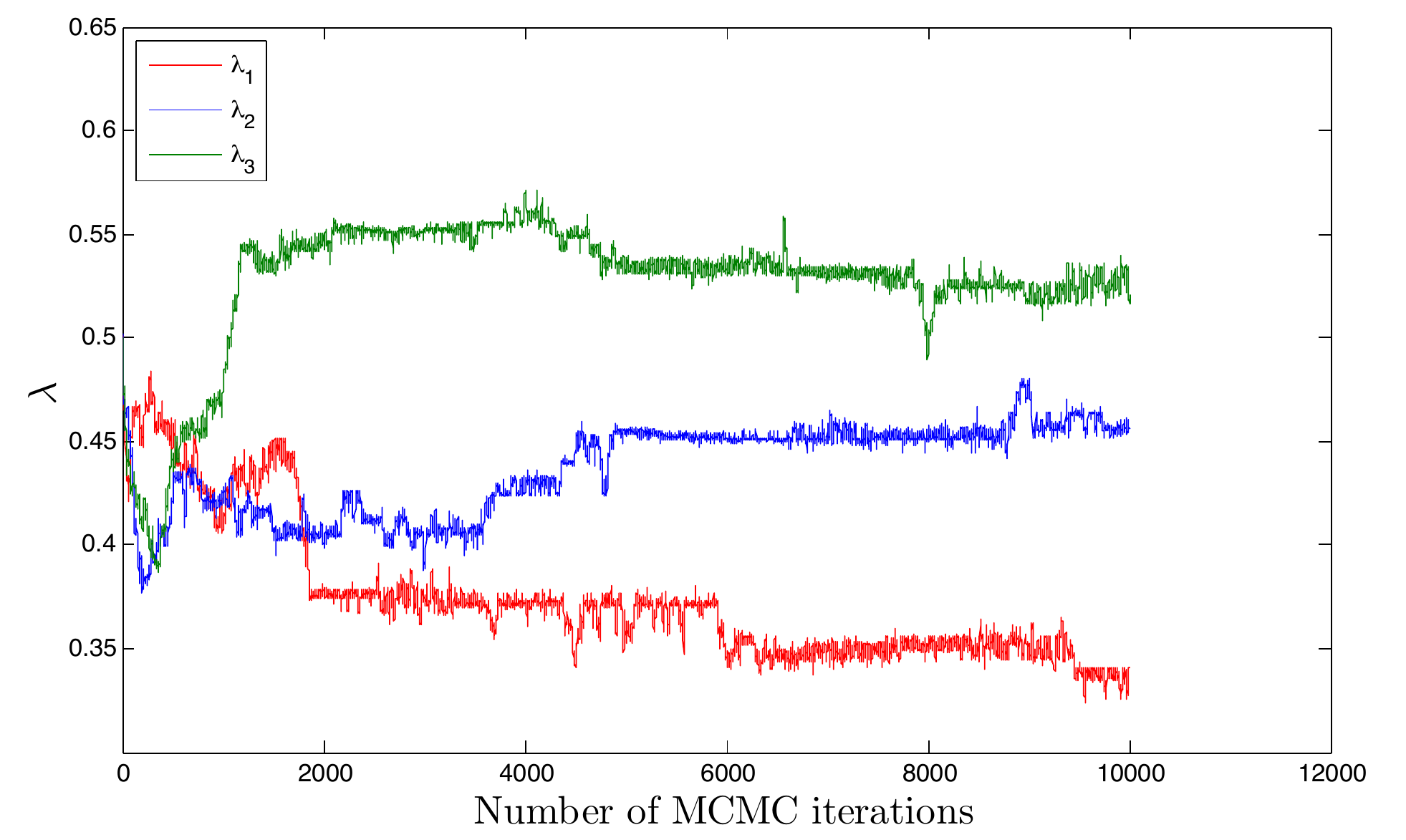}\\
		\mbox{(A) Histograms for posterior samples of}\  \bm\lambda &  \mbox{(B) Traceplot for posterior samples of }\ \bm\lambda 
		\end{array}$
		\caption{Histoplots of $\bm\lambda$ posterior samples and traceplots for the related Markov chain. Results are generated from one simulation dataset with $n=3600, \sigma^2=4, m=50$. Initial prior specifications for $\lambda_1,\lambda_2$ and $\lambda_3$ are all Unif$(0.3, 1.25)$. }
	\end{center}
\end{figure}

\subsection{Data Application: The Autism Brain Imaging Data Exchange (ABIDE)}
We apply our method to the data from ABIDE, which is a consortium collecting and sharing resting-state fMRI data from 1,112 subjects. Covariate information such as age at scan, sex, IQ, handedness and diagnostic information are also available from ABIDE studies. Among the subjects, 539 individuals have Autism spectrum disorders (ASD), which are characterized by symptoms such as social difficulties, communication deficits, stereotyped behaviors and cognitive delays. The remaining subjects are the age-matched normal controls (NC). All the fMRI images are preprocessed through slice-timing, motion correction, nuisance signal regression and temporal filtering. The resulting fMRI data, which are $91\times109\times 91$ 3D matrices, are normalized and registered to Montreal Neurological Institute (MNI) 152 stereotactic space. We aim to investigate the voxel-wise measures of latent functional architecture of the brains through fractional amplitude of low-frequency fluctuations (fALFF)\citep{zou2008improved}. fALFF is a metric reflecting the percentage of power spectrum within low-frequency domain ($0.01-0.1$Hz) which characterizes the intensity of spontaneous brain activities. We calculate the fALFF for each subject at every voxel. Since the fALFF is restricted to $(0,1)$, we perform the following monotone transformation
\bg y_j(\bs) = \log\left(\frac{f_j(\bs)}{1-f_j(\bs)}\right),\ed
where $f_j(\bs)$ represents the fALFF for subject $j=1,\ldots,1112$ at brain location $\bs$ and treat the transformed data as our outcomes.

\begin{figure}[!t]
	\small
	\makebox[\textwidth][c]{\includegraphics[scale=0.40,angle=270]{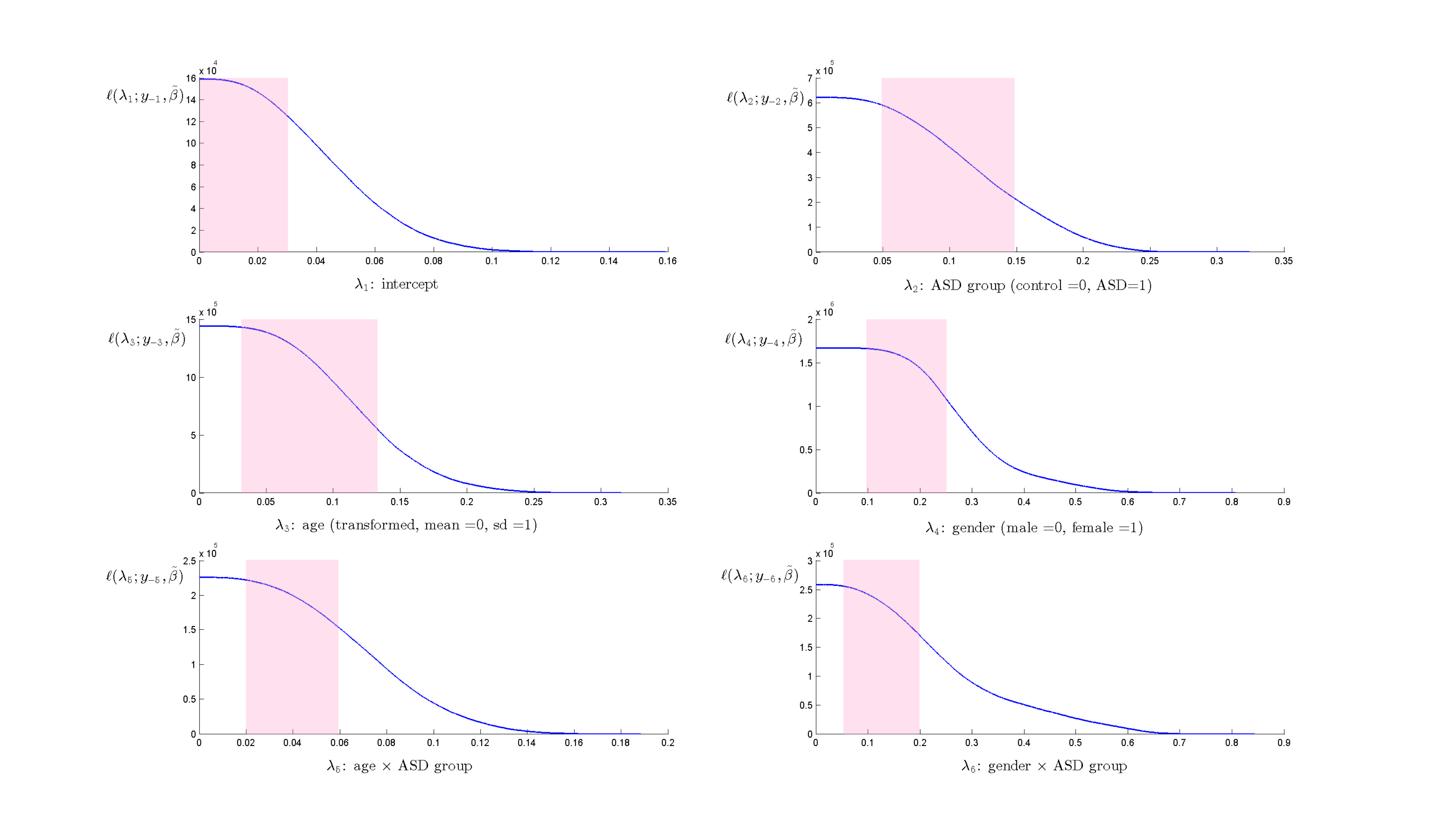}}%
	\caption{$\ell(\lambda_k)$ for specifying $\bm{\lambda}$ priors in the analysis of ABIDE data. The colored shades mark the intervals we choose as the range of the uniform priors for $\bm\lambda$.}\label{lmdselec}
\end{figure}

\begin{figure}[!t]
	\vskip 0.2in
	\begin{center}
		\small$
		\begin{array}{cc}
		\includegraphics[scale=0.35]{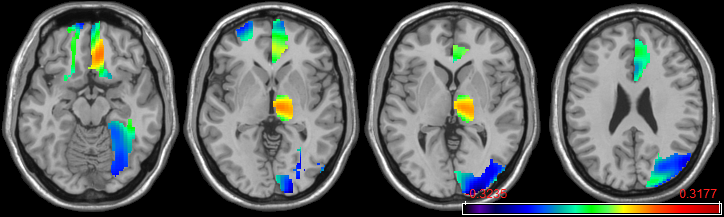}\\
		\includegraphics[scale=0.35]{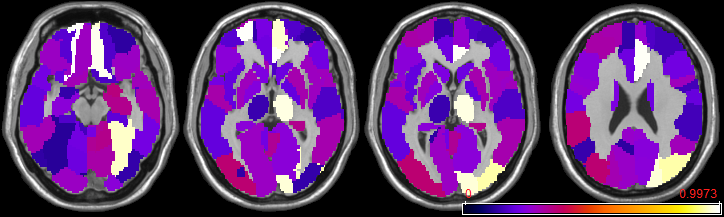}\\
		\mbox{(A) Covariate effects for the ASD group versus the control}\\ 
		\\
		\includegraphics[scale=0.35]{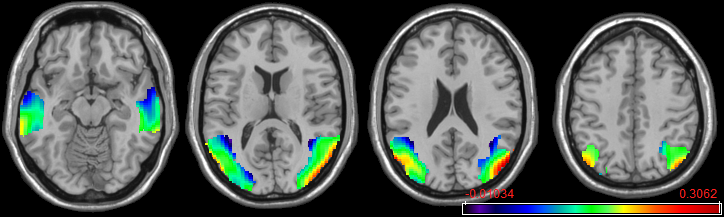}\\
		\includegraphics[scale=0.35]{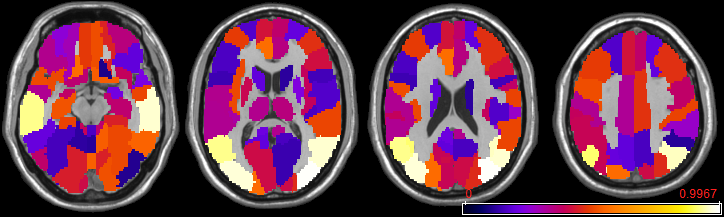}\\
		\mbox{(B) Covariate effects for the age}\\ \\
		\includegraphics[scale=0.35]{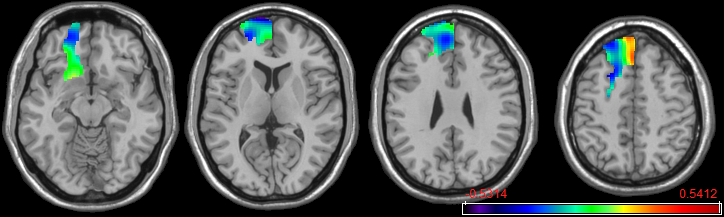}\\
		\includegraphics[scale=0.35]{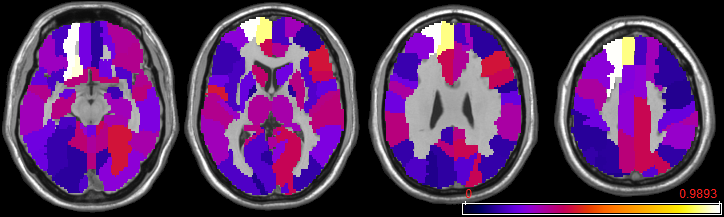}\\
		\mbox{(C) Covariate effects for group and gender interaction}
		\end{array}$
		\caption{\label{fig:res} Estimated SVCFs (top row in each subplot) and regional selection probabilities (bottom row in each subplot) based on posterior samples from our MCMC algorithm for ``ASD group", ``age" and ``ASD group$\times$ gender"}
	\end{center}
	\vskip -0.2in
\end{figure}


The covariates we choose for fitting model \eqref{model} are [1, group, age, gender, group$\times$ age, group$\times$ gender]. We use all the voxels at the gray matter as the observed spatial locations $\bs_1,\ldots,\bs_n$ and all the anatomical parcellation based on MNI templates as our brain regions $\mR_1,\ldots,\mR_G$ ($n=177,743 $ and $G=116$). The imaging outcomes are centered across all subjects at each voxel. The group variable equals to $1$ for the ASD subjects; the ages are all centered and scaled with zero mean and unit variance; the gender variable equals to $1$ for female subjects. The priors for the thresholding parameters are determined through the method described in subsection \ref{hyperP}. The profiles of $\widehat\ell(\lambda_k)$ are shown in Figure \ref{lmdselec}. The priors for the six thresholding parameters are Unif$(0,0.03)$, Unif$(0.05,0.15)$, Unif$(0.03,0.13)$, Unif$(0.1, 0.25)$, Unif$(0.02,0.06)$ and Unif$(0.05,0.2)$ according to the plots. The Gaussian kernel we use is the modified square exponential kernel with $a=0.25, b=95$ (bandwidth fixed without updating). To achieve $90\% $ recovery rate of the KL expansion, we set $L= 1,140$ eigenfunctions. The MCMC algorithm runs $60,000$ iterations with  $25,000$ burn-in.

Based on our results, the ASD subjects tend to show lower fALFF outcomes at the median and superior part of the right occipital lobe, which is the visual processing centers of human brains (the visual cortex). We observe significantly higher activities at the right fusiform gyrus, which has been reported to be related to Autism in \citep{hadjikhani2004activation}. Similar findings are observed at the right median orbitofrontal cortex, the region involved in most human cognition processes, especially decision-making, indicating more spontaneous brain cognition activities among the ASD subjects. From the axial view, Figure \ref{fig:res}(A) shows the information discussed above. Some other regions that are selected includes the right thalamus and the right anterior cingulum, which we do not discuss here in detail. 

Another major findings are the age effect on the fALFFs. We identified three brain regions that show higher fALFF outcomes as the age increases: the median occipital lobe, the median temporal lobe and the angular gyrus. These regions are generally involved in brain functions such as spatial temporal cognition, language, memory, attention and visual processing. Figure \ref{fig:res}(B) shows the findings above in brain slices from the axial view. 

Although no specific regions of interest are observed for the ``gender" variable, certain brain regions demonstrate different activation patterns its interaction with the disease group. Specifically, female ASD subjects have higher fALFFs as compared with male ASD effects at the left median and superior part of the orbital gyrus but lower fALFFs at the left frontal lobe gyrus and the left rectus. Figure \ref{fig:res}(C) shows these findings in three views for the ease of demonstration. Beyond these findings, we also note that the right inferior temporal gyrus displays smaller effects among the female ASD subjects as compared with the male autistics.

\section{Discussion}
In this paper, we introduced a new family of prior, the TMGP prior, for feature selections within spatially varying coefficient functions and its applications to massive neuroimaging data analysis. We demonstrate the prior large support properties of the TMGP prior and its posterior consistency under the spatially varying coefficient models under the spatial infill asymptotics. Simulation studies show that the TMGP prior is especially useful for imaging feature selections with relatively large noise or small sample sizes. 

In most spatial statistics literatures such as \citet{diggle1998model, gelfand2003spatial, smith2002predicting}, a spatial process are decomposed into three parts: a deterministic trend process or, in other words, mean process; a zero-mean variance process with continuous sample path and a zero-mean white noise process, i.e., the nugget effect. Under this general framework, \citet{zhu2014spatially} considered a more complex model as compared with our model \eqref{model}, which could be expressed using our notations as
\bg \label{modelzhu} y_j(\bs) = \bx_{j}^\rT\bbeta(\bs) + \eta_j(\bs)+ e_j(\bs), \ed
for all subjects $j=1,...,m$. They estimated the additional term $\eta_j(\bs)$ for every subject through standard local linear regression techniques, which do not require a pre-specified spatial covariance structure $\bm\Sigma_\eta$, or equivalently, its Karhunen-Lo\`{e}ve basis. Since in this paper, our primary focus is on the GLM framework for imaging data, we did not apply our TMGP prior under the setting of model \eqref{modelzhu}. To enable a similar analysis using the TMGP prior for $\bbeta(\bs)$ in \eqref{modelzhu} under the Bayesian framework, there are four major tasks. First, a proper prior specification for $\eta_j(\bs)$ needs to be introduced which should be flexible enough to capture various spatially smooth dynamics. Second, a computationally efficient algorithm for estimating the additional parameters brought by $\eta_j(\bs)$ is required since this set of parameters scale with the number of subjects. Third, since for each subject, we will have a subject specific random effect term, we need to carefully monitor the model fitting procedures to avoid potential over-fitting issues. Fourth, the theoretical analysis for posterior consistency in Theorem \ref{poscons} needs to be adapted to the more challenging model structure.

In addition to applying the TMGP prior to model \eqref{modelzhu}, our study can be extended to some other directions. With a focus on neuroimaging studies using model \eqref{model}, we can extend the prior constructions to enable self-guided parcellation while conducting feature selection tasks. This can help relax the region based sparsity assumptions for the SVCFs. We can also explore the Bayesian asymptotic theories when the number of brain region partitions diverges. With a focus on general Bayesian analysis, we can extend TMGP for modeling high-dimensional multivariate binary processes or selecting features for scalar-on-image models such as neuropsychiatric disease predictions. 

\begin{appendices}
	
\section{Proof of Theorem \ref{largsup}}
	Based on the assumptions for $\beta^0(\bs)$, let $I_1=I_1[\beta^0(\bs)]$ and $I_0=\{1,...,G\}\backslash I_1$, we have that
	\begin{align}
	\label{ineq1}
	\Pi\big(\|\beta(\bs)-\beta^0(\bs)\|_\infty&<\varepsilon\big)\geq\nonumber\\ &\Pi\left(\sup_{\bs\in\cup_{g\in I_1}\mR_g}|\tbeta(\bs)-\beta^0(\bs)|<\varepsilon, \inf_{\bs\in\cup_{g\in I_1}\mR_g}|\tbeta(\bs)|>\lambda, \sup_{\bs\in\cup_{g\in I_0}\mR_g}|\tbeta(\bs)|\leq \lambda\right).
	\end{align}

	Without loss of generality, we only consider $0<\varepsilon<\lambda_0-\lambda$. Note that for all $\bs\in\cup_{g\in I_1}\mR_g$, $|\tbeta(\bs)-\beta^0(\bs)|<\varepsilon$ and $|\beta^0(\bs)|\geq \lambda_0$ implies that $|\tbeta(\bs)|\geq \lambda_0-\varepsilon>\lambda$, then \eqref{ineq1} is equivalent to
	$$ \Pi\left(\|\beta(\bs)-\beta^0(\bs)\|_\infty<\varepsilon\right)\geq \Pi\left(\sup_{\bs\in\cup_{g\in I_1}\mR_g}|\tbeta(\bs)-\beta^0(\bs)|<\varepsilon, \sup_{\bs\in\cup_{g\in I_0}\mR_g}|\tbeta(\bs)|\leq \lambda\right).$$
	
	Let $\phi_l(\bs)$ and $\zeta_l, l=1,...,\infty,$ be the normalized eigenfunctions and eigenvalues of the kernel function $\kappa(\cdot, \cdot)$, then the Karhunen-Lo\`{e}ve expansions of $\gamma(\bs)$ and $\epsilon(\bs)$ can be expressed as
	$\gamma(\bs)=\sum_{l=1}^\infty u_l\phi_l(\bs)$, $\forall \bs\in \mR,$ and $\epsilon(\bs) =\sum_{l=1}^\infty v_{lg}\phi_l(\bs)$, $\bs\in \mR_g, g=1,...,G$, such that $u_l\iid N(0, \zeta_l\tau^2)$, $v_{lg}\iid N(0, \zeta_l\theta^2)$ and $u_l, v_{lg}$ are all independent. Since the RKHS of $\kappa(\cdot, \cdot)$ is $\mC(\mR)$, for $\bs$ within any partition $\mR_g$,  $\beta^0(\bs)$ can be represented as $\sum_{l=1}^\infty w_{lg}\phi_l(\bs),$ where $\sum_{l=1}^\infty w_{lg}^2<\infty$.
	
	For $\bs\in \mR_g$ with $g\in I_1$
	\bg \label{inequ2} \sup_{\bs\in\mR_g}|\tbeta(\bs)-\beta^0(\bs)|\leq \sup_{\bs\in\mR_g}|\tbeta_{L,g}(\bs)-\beta^0_{L,g}(\bs)|+\sup_{\bs\in\mR_g}|\tbeta_{L,g}^*(\bs)|+\sup_{\bs\in\mR_g}|\beta_{L,g}^{0*}(\bs)|,\ed
	where $\tbeta_{L,g}(\bs)=\sum_{l=1}^{L}(u_l+v_{lg})\phi_l(\bs),\ \beta^0_{L,g}(\bs)=\sum_{l=1}^{L} w_{lg}\phi_l(\bs),\ \tbeta_{L,g}^*(\bs)=\tbeta(\bs)-\tbeta_{L,g}(\bs)$ and $\beta_{L,g}^{0*}(\bs)=\beta^0(\bs)-\beta^0_{L,g}(\bs)$.
	Since the RKHS of $\kappa(\cdot, \cdot)$ is $\mC(\mR)$, $\tbeta(\bs)$ is uniformly continuous on $\mR_g$ with probability $1$, then by Theorem 3.1.2 of \citet{adler2009random}, $\lim_{L\to\infty}\sup_{\bs\in\mR_g}|\tbeta_{L,g}^*(\bs)|=0$ with probability $1$. By the uniform convergence of the series $\sum_{l=1}^L w_{lg}\phi_l(\bs)$ to $\beta^0(\bs)$ as $L\to\infty$ on $\mR_g$, $\lim_{L\to\infty}\sup_{\bs\in\mR_g}|\beta_{L,g}^{0*}(\bs)|=0$. Then we can find a finite integer $L_g$ such that for all $L\geq L_g$,  $\sup_{\bs\in\mR_g}|\tbeta_{L,g}^*(\bs)|<\frac{\varepsilon}{3}$ with probability $1$ and $\sup_{\bs\in\mR_g}|\beta_{L,g}^{0*}(\bs)|<\frac{\varepsilon}{3}$. Since $\phi_l(\bs), l=1,...,L_g$ are all continuous functions in on $\mR$, we have that $\max_{1\leq l\leq L_g}\|\phi_l(\bs)\|_\infty<M_{\phi, L_g}$ where $M_{\phi, L_g}$ is a certain constant. Let $|u_l+v_{lg}-w_{lg}|<\frac{\varepsilon}{3L_g M_{\phi, L_g}}$ for all $l=1,...,L_g$ and consider $L=L_g$ in \eqref{inequ2}, we have that $\sup_{\bs\in\mR_g}|\tbeta_{L,g}(\bs)-\beta^0_{L,g}(\bs)|<\frac{\varepsilon}{3}$. Thus, the condition $|u_l+v_{lg}-w_{lg}|<\frac{\varepsilon}{3L_g M_{\phi, L_g}}, l=1,...,L_g$ can guarantee that $ \sup_{\bs\in\mR_g}|\tbeta(\bs)-\beta^0(\bs)|<\varepsilon$ with probability $1$ for $g \in I_1$.
	
	For $\bs\in \mR_g$ with $g\in I_0$, similar to \eqref{inequ2} and the definitions above, we have
	\bg \label{inequ3} \sup_{\bs\in\mR_g}|\tbeta(\bs)|\leq \sup_{\bs\in\mR_g}|\tbeta_{L,g}(\bs)|+\sup_{\bs\in\mR_g}|\tbeta_{L,g}^*(\bs)|.\ed
	Similarly, we can find $L_g$ and $M_{\phi, L_g}$ such that $|u_l+v_{lg}|\leq \frac{\lambda}{2L_g M_{\phi, L_g}}, l=1,...,L_g$ guarantees that $\sup_{\bs\in\mR_g}|\tbeta(\bs)|\leq\lambda$ with probability $1$ for all $g\in I_0$. 
	
	Then we have
	\begin{eqnarray}\label{sets} \Pi\left(\|\beta(\bs)-\beta^0(\bs)\|_\infty<\varepsilon\right)\geq & \Pi\bigg(\left\{|u_l+v_{lg}-w_{lg}|<\frac{\varepsilon}{3L_g M_{\phi, L_g}} : l=1,...,L_g, g\in I_1\right\}\cup\nonumber \\
	&\quad\left\{ |u_l+v_{lg}|\leq \frac{\lambda}{2L_gM_{\phi, L_g}}: l=1,...,L_g, g\in I_0\right\} \bigg)>0,
	\end{eqnarray}
	due to the positive measures assigned on arbitrary nonempty sets by the $\left(\sum_{g=1}^G L_g + L_\max \right)$-dimensional multivariate Gaussian distribution: $\left(u_1,...,u_{L_\max}, v_{11},...,v_{L_11},..., v_{1G},...,v_{L_GG}\right)$, where $L_\max= \max_{g=1,...,G} L_g $.

\section{Proof of Theorem \ref{poscons}}

\subsection{KL neighborhood conditions for noniid outcomes}
\begin{lemma}\label{klnbr}
	Consider our observed data $\by(\bs_i)\in \mathbb{R}^m, \ \by(\bs_i)\sim f_{i, \bbeta}(\by) $ where $$f_{i, \bbeta}(\by)=(2\pi\sigma^2)^{-m/2}\exp\left\{-\frac{1}{2\sigma^2}\|\by-\bX\bbeta(\bs_i)\|_2^2\right\}$$ for some constant $\sigma^2>0$. Define 
	$$D_i(\bbeta_0, \bbeta) = \log\frac{f_{i, \bbeta_0}}{f_{i, \bbeta}},
	\quad K_i(\bbeta_0, \bbeta)= E_{f_{i, \bbeta_0}} [D_i(\bbeta_0, \bbeta)],
	\quad V_i(\bbeta_0, \bbeta) = Var_{f_{i, \bbeta_0}} [D_i(\bbeta_0, \bbeta)].$$
	If we a assign an independent TMGP prior for each dimension of $\bbeta$, i.e., $$\beta_k(\bs)\sim \TMGP(\tau^2_k, \theta^2_k, \lambda_k, \kappa(\cdot, \cdot)),$$ then we have that $\exists B, \Pi(B)>0$ such that 
	$$\liminf_{n\to\infty} \Pi\left(\left\{\bbeta\in B: \frac{1}{n}\sum_{i=1}^{n}K_i(\bbeta^0, \bbeta) <\varepsilon\right\} \right)>0, $$
	$$\frac{1}{n^2}\sum_{i=1}^n V_i(\bbeta^0, \bbeta)\to 0, \forall \bbeta\in B.$$
	
\end{lemma}
\begin{proof}
	It is trivial to have that $$D_i(\bbeta_0, \bbeta) = \frac{1}{\sigma^2} (\bbeta^0(\bs_i)-\bbeta(\bs_i))^\top\left[\bX^\top\by -\frac{1}{2}\bX\bX^\top(\bbeta^0(\bs_i)+\bbeta(\bs_i))\right],$$
	then since $E_{f_{i, \bbeta_0}}[\by]=\bX\bbeta(\bs_i), Var_{f_{i, \bbeta_0}}[\by]=\sigma^2\bm I_m $, 
	$$K_i(\bbeta_0, \bbeta) =\frac{1}{2\sigma^2}(\bbeta^0(\bs_i)-\bbeta(\bs_i))^\top\bX^\top\bX(\bbeta^0(\bs_i)-\bbeta(\bs_i))\leq \frac{m d_\max }{2\sigma^2}\|\bbeta^0(\bs_i)-\bbeta(\bs_i)\|_2^2,$$
	$$V_i(\bbeta_0, \bbeta)=\frac{1}{\sigma^2}(\bbeta^0(\bs_i)-\bbeta(\bs_i))^\top\bX^\top\bX(\bbeta^0(\bs_i)-\bbeta(\bs_i))\leq \frac{m d_\max }{\sigma^2}\|\bbeta^0(\bs_i)-\bbeta(\bs_i)\|_2^2.$$
	Now consider $B_k = \left\{\beta_k(\bs): \|\beta_k(\bs)-\beta^0(\bs)\|_\infty <\sqrt{\frac{2\sigma^2\varepsilon}{mp d_\max}}\right\}$ and let $B=\cap_{k=1}^p B_k$. Since the priors for $\beta_k(\bs), k=1,...,p$ are independent and $\Pi(B_k)>0$ due to Theorem \eqref{largsup}, $\Pi(B)=\prod_{k=1}^p \Pi(B_k)>0$.
	
	For all $\bbeta\in B$, $K_i(\bbeta_0, \bbeta)\leq \frac{m d_\max }{2\sigma^2}\sum_{k=1}^p\|\beta^0_k(\bs)-\beta_k(\bs)\|_\infty^2 <\varepsilon$ and similarly $V_i(\bbeta_0, \bbeta)<2\varepsilon$. Then $\liminf_{n\to\infty} \Pi\left(\left\{\bbeta\in B: \frac{1}{n}\sum_{i=1}^{n}K_i(\bbeta^0, \bbeta) <\varepsilon\right\} \right)=\Pi(B)>0$ and $\frac{1}{n^2}\sum_{i=1}^n V_i(\bbeta^0, \bbeta)<\frac{2\varepsilon}{n}\to 0$ for all $\bbeta\in B$.
\end{proof}

\subsection{Sieve constructions}

Define the set of functions 
$$\mP_n=\left\{\beta(\bs)\in\mP: \|\beta(\bs)\|_\infty <\sqrt{n}, \ \sup_{\bs\in \mR_g}\left|D^{\alpha}\beta(\bs)\right|<\sqrt{n}, \ g\in I_1[\beta(\bs)],\ 1\leq \|\alpha\|_1 \leq \rho \right\},$$
as our sieve construction.

\begin{lemma}\label{covering}
	If $G<\infty$, the $\varepsilon$-covering number under the sup-norm for $\mP_n$ satisfies $$\log N(\varepsilon, \mP_n, \|\cdot\|_\infty) <Cn^{\frac{d}{2\rho}}\varepsilon^{-d},$$
	for some finite constant $C$.
\end{lemma}
\begin{proof}
	Define $$\mP_{n,g}=\left\{\beta(\bs)\in\mC^\rho(\overline\mR_g): \sup_{\bs\in \mR_g}|D^\alpha\beta(\bs)| <\sqrt{n}, \ 0\leq \|\alpha\|_1\leq \rho \right\},$$ for all $g=1,...,G$. Theorem 2.7.1 of \citet{van1996weak} implies that 
	$$\log N(\varepsilon, \mP_{n,g} , \|\cdot\|_\infty)\leq C_gn^{\frac{d}{2\rho}}\varepsilon^{-d},$$
	for some constants $C_g<\infty$.  Then by the definition of $\mP_n$, we have that
	$$N(\varepsilon, \mP_n, \|\cdot\|_\infty) \leq \prod_{g=1}^G N(\varepsilon, \mP_{n,g} , \|\cdot\|_\infty)\leq \exp\left\{Cn^{\frac{d}{2\rho}}\varepsilon^{-d}\right\},$$ 
	where $C=\sum_{g=1}^G C_g<\infty$.
\end{proof}

\begin{lemma}\label{seivemr}
	Consider the TMGP prior for $\beta(\bs)$ with kernel function satisfying condition (K1)(K2), then $\Pi(\mP\cap\mP_n^c)\leq De^{-bn}$ for some constant $D, b>0$.
\end{lemma}
\begin{proof}
	The construction of the TMGP prior implies that 
	$$\Pi(\mP_n) \geq \prod_{g=1}^G \Pi\left(\sup_{s\in\mR_g}|D^\alpha \tbeta_g(\bs)|>\sqrt{n}, 0\leq\|\alpha\|_1\leq \rho\right),$$
	where $\tbeta_g(\bs)\iid \GP(0, (\theta^2+\tau^2)\kappa(\bs, \bs'))$ for all $g=1,...,G$. By applying Theorem 5 of \citet{ghosal2006posterior}, we have that 
	$\Pi\left(\sup_{s\in\mR_g}|D^\alpha \tbeta_g(\bs)|>\sqrt{n}, 0\leq\|\alpha\|_1\leq \rho\right)\geq 1-Ae^{-bn}$ for some $A, b>0$, given that $\kappa(\bs, \cdot)$ has continuous partial derivatives of order $2\rho+2$ on the compact set $\mR$. Then we have that $\Pi(\mP\cap\mP_n^c)\leq 1-(1-Ae^{-bn})^G \leq De^{-dn}$ where $D=AG$ due to the fact that $(1-x)^G \geq 1-Gx$ for all $0<x<1$ and $G=1,2,...$.
\end{proof}

Now we define $\textbf{P}_n=\{\bbeta(\bs)=[\beta_1(\bs),...,\beta_p(\bs)]^\top: \beta_k(\bs)\in \mP_n, k=1,...,p\}$ here and below. Then we can easily get that 
\bg\label{sievecover}
N(\varepsilon, \textbf{P}_n, \|\cdot\|_\infty)<  \exp\{Cpn^{\frac{d}{2\rho}}\varepsilon^{-d}\},\ed
and if assign TMGP priors independently for all elements in $\bbeta(\bs)$ then
\bg\label{sievesize} \Pi\left(\textbf{P}\cap\textbf{P}_n^c\right)\leq Dp\exp\{-bn\}.\ed

\subsection{Test Constructions}

\begin{lemma}\label{vecbound}
	Consider $\by\sim N(\bX\bbeta, \sigma^2\bm I_m)$, a standard linear model with sample size $m$ where $\by=[y_1,...,y_m]$; $\bX$ is an $m\times p$ design matrix satisfying assumption (X1). Consider the test function $\Phi=I\left(\|\hat{\bbeta}-\bbeta^0\|_2 >\frac{\varepsilon \sqrt p}{2}\right)$ for testing $H_0: \bbeta=\bbeta^0  \ versus \ H_1: \bbeta=\bbeta^1,$
	where $\bbeta^0\in \mathbb{R}^p$ and $\bbeta^1\in \{\bbeta \in \mathbb{R}^p: \|\bbeta -\bbeta^0\|_\infty\geq\varepsilon\}$; $\hat\bbeta =(\bX^\top\bX)^{-1}\bX\by$ is the ordinary least square estimator. Then for $m>\frac{8(\log2)\sigma^2}{\varepsilon^2d_\min}$, we have that 
	$$E_{P_0}[\Phi]\leq \exp\{-\Omega_m mp\}, \quad E_{P_1}[1-\Phi] \leq \exp\{-\Omega_m mp\},$$
	for some $\Omega_m>0$ depending on $m$, where $P_0$ and $P_1$ represents the probability distributions under $H_0$ and $H_1$.
\end{lemma}

\begin{proof}Note that for $t=0, 1$, under $H_t$, $\|\hat{\bbeta}-\bbeta^t\|_2^2  d_\min m /\sigma^2\leq (\hat{\bbeta}-\bbeta^t)^\top\bX^{T}\bX(\hat{\bbeta}-\bbeta^t)/\sigma^2\sim \chi_p^2$, then we have that
	\bg E_{P_0}[\Phi]=P_0\left(\|\hat{\bbeta}-\bbeta_0\|_2^2 >\frac{\varepsilon^2p}{4}\right)\leq P_0\left(\chi_p^2>\frac{\varepsilon^2p d_\min m}{4\sigma^2}\right)\leq \exp\left\{-\left(\frac{\varepsilon^2 d_\min }{16\sigma^2}-\frac{\log 2}{2m}\right)mp\right\},\ed
	where the last inequality is simply due to the f act that $P(\chi_p^2 > x)\leq (1-2t)^{-p/2}\exp\{-tx\}, \forall 0<t<1/2$ by letting $t=1/4$. Similarly, 
	\begin{eqnarray}
	E_{P_1}[1-\Phi] &=& P_1\left(\|\hat{\bbeta}-\bbeta^0\|_2\leq\frac{\varepsilon\sqrt p}{2}\right)\nonumber\\
	&\leq& P_1\left(\left|\|\hat{\bbeta}-\bbeta^1\|_2 - \|\bbeta^0-\bbeta^1\|_2\right | \leq\frac{\varepsilon\sqrt p}{2}\right)\nonumber\\
	&\leq&  P_1\left(\|\hat{\bbeta}-\bbeta^1\|_2 \geq -\frac{\varepsilon \sqrt p}{2} +  \|\bbeta^0-\bbeta^1\|_2 \right)\nonumber \\
	&\leq & P_1\left(\|\hat{\bbeta}-\bbeta^1\|_2 \geq \frac{\varepsilon \sqrt p}{2}  \right)\nonumber\\
	&\leq &\exp\left\{-\left(\frac{\varepsilon^2 d_\min }{16\sigma^2}-\frac{\log 2}{2m}\right)mp\right\}.
	\end{eqnarray}
	Define $\Omega_m=\frac{\varepsilon^2 d_\min}{16\sigma^2}-\frac{\log 2}{2m}$ here and below. Notice that $m>\frac{8(\log2)\sigma^2}{\varepsilon^2d_\min}$ is equivalent to $\Omega_m>0$, we complete the proof.
\end{proof}

\begin{lemma}\label{matrixbound}
	We consider $n_0$ locations $\bs_1,...,\bs_{n_0}$ and define $\bbeta_i=[\beta_1(\bs_i),...,\beta_p(\bs_i)]^\top$ ($\bbeta_i^0$ and $\bbeta_i^1$ can be defined accordingly). Suppose that we have observed the data $\by_i = [y_1(\bs_i),...,y_m(\bs_i)]^\top$ generated from the SVCM, i.e.  $\by_i\sim N(\bX\bbeta_i, \sigma^2\bm I_m)$. Consider testing $H_0: \bbeta_i=\bbeta_i^0, i=1,...,n_0 \ versus \ H_1: \bbeta_i=\bbeta_i^1, i=1,...,n_0$ where $\|\bbeta_i^1 -\bbeta_i^0\|_\infty\geq\varepsilon$ for all $i=1,...,n_0$. Define $\Phi_i=I\left(\|\hat{\bbeta}_i-\bbeta_i^0\|_2 >\frac{\varepsilon \sqrt p}{2}\right)$ with $\hat{\bbeta}_i=(\bX^\top\bX)^{-1}\bX\by_i$. Then for the test function $\tilde\Phi=I\left(\sum_{i=1}^{n_0} \Phi_i >\frac{n_0}{2}\right) $, we have that
	$$E_{P_0}[\tilde\Phi]\leq  \exp\{-Cn_0\},\quad  E_{P_1}[1-\tilde\Phi]\leq \exp\{-Cn_0\} $$
\end{lemma}

\begin{proof}
	By the results from Lemma \ref{vecbound}, $E_{P_0}[\Phi_i]\leq e^{-\Omega_m mp}$ and $E_{P_1}[1-\Phi_i]\leq e^{-\Omega_m mp}$ for all $i=1,...,n_0$. Then 
	\bg \label{tp1}E_{P_0}[\tilde\Phi] \leq  P_0\left(\sum_{i=1}^{n_0} \Phi_i - \sum_{i=1}^{n_0}E_{P_0}[ \Phi_i] > \frac{n_0}{2}-n_0  e^{-\Omega_m mp} \right),\ed
	
	\bg \label{tp2}E_{P_1}[1-\tilde\Phi] = P_1\left(\sum_{i=1}^{n_0} (1-\Phi_i)\geq \frac{n_0}{2}\right)\leq P_1\left(\sum_{i=1}^{n_0} (1-\Phi_i)- \sum_{i=1}^{n_0} E[1-\Phi_i]\geq \frac{n_0}{2}-n_0  e^{-\Omega_m mp} \right).\ed
	By the Hoeffding inequality \citep{hoeffding1963probability}, the right hand side of both \eqref{tp1} and \eqref{tp2} are bounded by $\exp\left\{-\frac{2}{n_0} \frac{n_0^2(1-2e^{-\Omega_m mp} )^2}{4}\right\}$ when $1-2e^{-\Omega_m mp}>0$. That is,
	$E_{P_0}[\tilde\Phi]\leq  \exp\{-Cn_0\},\  E_{P_1}[1-\tilde\Phi]\leq \exp\{-Cn_0\} $ where $C=\frac{(1-2e^{-\Omega_m mp} )^2}{2}$.
\end{proof}

\begin{lemma}\label{point}
	For two functions $\beta^0(\bs), \beta^1(\bs)\in\mP$, if $\|\beta^0(\bs)-\beta^1(\bs)\|_1 =\int_{\bs\in\mR}|\beta^0(\bs)-\beta^1(\bs)|\bP_n(d\bs)\geq\varepsilon$, we have that 
	$\bP_n(|\beta^0(\bs)-\beta^1(\bs)|\geq\frac{\varepsilon}{2} ) \geq c'$ where $0<c'\leq1$ is a constant. That is, the set $\{\bs\in\{\bs_1,...,\bs_n\}: |\beta^0(\bs)-\beta^1(\bs)|\geq\frac{\varepsilon}{2} \}$ has $n_0\geq c'n-1$ elements.
\end{lemma}
\begin{proof}
	Let $\mathcal S =\{\bs\in \mR: |\beta^0(\bs)-\beta^1(\bs)|\geq\frac{\varepsilon}{2}\}$, then 
	\begin{eqnarray}
	\varepsilon&\leq &\int_{\bs\in\mR}|\beta^0(\bs)-\beta^1(\bs)|\bP_n(d\bs)\nonumber\\
	& =& \int_{\bs\in\mathcal S}|\beta^0(\bs)-\beta^1(\bs)|\bP_n(d\bs) + \int_{\bs\in \mR\backslash\mathcal S}|\beta^0(\bs)-\beta^1(\bs)|\bP_n(d\bs)\nonumber\\
	&\leq& (M_0+M_1)\bP_n\left(|\beta^0(\bs)-\beta^1(\bs)|\geq\frac{\varepsilon}{2}\right )  +
	\frac{\varepsilon}{2}\bP_n(\mR)\nonumber,
	\end{eqnarray}
	where $\bP_n(\mR)=1$; $M_0=\|\beta^0(\bs)\|_\infty$ and $M_1=\|\beta^1(\bs)\|_\infty$ are finite constants due to absolute continuity. Thus $\bP_n\left(|\beta^0(\bs)-\beta^1(\bs)|\geq\frac{\varepsilon}{2}\right ) \geq c'$ by letting $c'=\frac{\varepsilon}{2(M_0+M_1)}$.
\end{proof}

\begin{lemma}\label{matrixfuncbound}
	There exists a test $\Phi_{\bbeta^1, \bbeta^0}$ for testing $H_0: \bbeta(\bs)=\bbeta^0(\bs)$ against $H_1: \bbeta(\bs)=\bbeta^1(\bs)$ where $\|\bbeta^1(\bs)-\bbeta^0(\bs)\|_{1, \infty}\geq \varepsilon$ in our proposed SVCM, such that
	$$E_{P_0}[\Phi_{\bbeta^1, \bbeta^0}]\leq \exp\{-Cn\}, \quad E_{P_1}[1-\Phi_{\bbeta^1, \bbeta^0}]\leq \exp\{-Cn\}, $$
	for some constant $C$ with $P_0$ and $P_1$ corresponding to the probability distributions under $H_0$ and $H_1$.
\end{lemma}
\begin{proof}
	For two vector-valued functions $\bbeta^t(\bs)=[\beta^t_1(\bs),...,\beta^t_p(\bs)]^\top, t=0,1$, if $\|\bbeta^1(\bs)-\bbeta^0(\bs)\|_{1, \infty}\geq \varepsilon$, we must have at least one $k\in\{1,...,p\}$, such that $\|\beta^1_k(\bs)-\beta^0_k(\bs)\|_1 \geq \varepsilon$, then due to Lemma \ref{point}, we can find $n_0\geq c'n-1$ elements in $\{\bs_1,...,\bs_n\}$ such that $|\beta^1_k(\bs)-\beta^0_k(\bs)|\geq\frac{\varepsilon}{2}$. Without loss of generality, we denote these points as $\bs_1,...,\bs_{n_0}$. Then for all $\bs_i, i=1,...,n_0$, we have that $\|\bbeta^1(\bs_i)-\bbeta^0(\bs_i)\|_\infty\geq \frac{\varepsilon}{2}$. 
	
	Now define the set $\mS_{\bbeta^1, \bbeta^0}=\{\bs\in\{\bs_1,...,\bs_n\}: \|\bbeta^1(\bs_i)-\bbeta^0(\bs_i)\|_\infty\geq \frac{\varepsilon}{2}\}$. Then $n_0=|\mS_{\bbeta^1, \bbeta^0}|\geq c'n-1$. Define the test function
	$$\Phi_{\bbeta^1, \bbeta^0}=I\left( \sum_{\bs\in \mS_{\bbeta^1, \bbeta^0}} \Phi(\bs) >\frac{n_0}{2}\right),$$
	where $\Phi(\bs)=I\left(\|(\bX^\top\bX)^{-1}\bX\by(\bs)-\bbeta^0(\bs)\|_2 >\frac{\varepsilon \sqrt p}{2}\right)$. Then by Lemma \ref{matrixbound} (replacing $\varepsilon$ by $\varepsilon/2$) we have
	$$E_{P_0}[\Phi_{\bbeta^1, \bbeta^0}]\leq \exp\{-C_0n_0\},\quad E_{P_1}[\Phi_{\bbeta^1, \bbeta^0}]\leq \exp\{-C_0n_0\},$$
	where $C_0>0$ is a constant. Since $n_0\geq c\prime n -1 $ for a positive constant $c^\prime$, we have that $E_{P_0}[\Phi_{\bbeta^1, \bbeta^0}]\leq \exp\{-Cn\}$ and $E_{P_1}[\Phi_{\bbeta^1, \bbeta^0}]\leq \exp\{-Cn\}$.
\end{proof}

\begin{lemma}\label{matrixfuncboundall}
	There exists a test $\Psi$ for testing $H_0: \bbeta(\bs)=\bbeta^0(\bs)$ against $H_1: \bbeta(\bs)\in U_{\varepsilon,n}^c=U_{\varepsilon}^c\cap \textbf{P}=\{\bbeta(\bs)\in \textbf P_n: \|\bbeta(\bs)-\bbeta^0(\bs)\|\geq \varepsilon\}$ in our proposed SVCM, such that
	$$E_{P_0}[\Psi]\leq \exp\{-d_0n\}, \quad E_{P_1}[1-\Psi]\leq \exp\{-d_1n\}, $$
	for some constant $d_0, d_1$ with $P_0$ and $P_1$ corresponding to the probability distributions under $H_0$ and $H_1$.
\end{lemma}
\begin{proof}
	Let $\mathcal N=N(\frac{\varepsilon}{2}, \textbf P_n, \|\cdot\|_\infty)$ be the covering number of $\textbf{P}_n$ by $\varepsilon/2$-balls under the supreme norm.  
	Then for all $\bbeta(\bs)\in U_{\varepsilon,n}^c$, we can find $\bbeta^j(\bs), j\in\{1,...,\mathcal N\}$ such that $\|\bbeta^j(\bs)-\bbeta(\bs)\|_\infty\leq \frac{\varepsilon}{2}$, which implies that $\|\bbeta^j(\bs)-\bbeta^0(\bs)\|_\infty\geq \|\bbeta^0(\bs)-\bbeta(\bs)\|_\infty- \|\bbeta^j(\bs)-\bbeta(\bs)\|_\infty\geq \frac{\varepsilon}{2}$ for all $j=1,...,\mathcal N$. Following the notations and results in Lemma \ref{matrixfuncbound} with regard to $\varepsilon/2$, we have that the tests $\Phi_{\bbeta^j, \bbeta^0}$ all satisfy that $E_{{\bbeta^0}}[\Phi_{\bbeta^j, \bbeta^0}]\leq \exp\{-d_1 n\}$ and $E_{{\bbeta^j}}[\Phi_{\bbeta^j, \bbeta^0}]\leq \exp\{-d_1 n\}$ for some constant $d_1$. Now for the test function $$\Psi=\max_{j=1,...,\mathcal{N}} \Phi_{\bbeta^j, \bbeta^0},$$ which only depend on the set $\textbf P_n$ instead of specific $\bbeta(\bs)$ in the alternative hypothesis, 
	$$E_{P_0}[\Psi]\leq \sum_{j=1}E_{{\bbeta^j}}[\Phi_{\bbeta^j, \bbeta^0}]\leq \mathcal{N}\exp\{-d_1 n\}<\exp\{Cp n^{\frac{d}{2\rho}}\varepsilon^{-d}-d_1 n\}\leq\exp\{-d_0n\},$$
	for some constant $d_0$ due to \eqref{sievecover} and the fact that $n^{\frac{d}{2\rho}}=o(n)$. At the same time
	$$ E_{P_1}[1-\Psi]\leq E_{{\bbeta^1}}[1-\Phi_{\bbeta^1, \bbeta^0}]\leq  \exp\{-d_1 n\},$$
	which complete our proof.
\end{proof}
Now based on Lemma \ref{klnbr}, equation \eqref{sievesize} and Lemma \ref{matrixfuncboundall}, Theorem \ref{poscons} follows from a direct application of Theorem A.1. of \citet{choudhuri2004bayesian}.

\section{Details about the MCMC algorithm}
We list the details about our MCMC algorithm here. Denote by $\phi(\cdot; \bm \mu, \bm{\Sigma})$ the density function of $N(\bm\mu, \bm\Sigma)$. We normally fix $v=w=0.001$ in the inverse-gamma priors and fix $\theta^2=1$ for the local GPs.
\begin{itemize}
\item Updating $\bm\tbeta(\bs_i), \ i=1,...,n$:
given the block structures of $\bm{\tbeta}(\bs_i)$, we update $\bm\tbeta_k^g=\{\tbeta_k(\bs_i)\}_{\bs_i\in\mR_g}, k=1,...,p, \ g=1,...,G$ separately with $p\times G$ M-H steps. Specifically, the full conditional $[\bm\tbeta_k^g\mid \bm \tbeta, \bm u, \sigma^2, \ \bm{\lambda}, \ \mathcal Y]$ is proportional to  
$$
h(\bm\tbeta_k^g)=\left[\prod_{i:\ \bs_i\in\mR_g}\prod_{j=1}^m\phi\left(y_{j, -k}(\bs_i);  x_{jk}\tbeta_k(\bs_i)I_{\lambda_k}[\tbeta_k(\bs_i)], \sigma^2\right)\right]\phi  \left(\bm\tbeta_k^g; \bm{\varphi}_g\bm u_k ,\theta^2 K_g \right),
$$
where $y_{j, -k}(\bs_i)=y_j(\bs_i)-\sum_{t\neq k}x_{jt}\tbeta_t(\bs_i)I_{\lambda_k}[\tbeta_k(\bs_i)]$. We adopt a Metropolis-Hasting (M-H) algorithm to update 
$\bm\tbeta_k^g$ by first generating a proposal, $\bm\tbeta_k^g+\Delta\bm\tbeta_k^g$ with a zero mean Gaussian fluctuation $\Delta\bm\tbeta_k^g$. Then we set $\bm\tbeta_k^g\leftarrow\bm\tbeta_k^g+\Delta\bm\tbeta_k^g$ with probability: $\min\left\{1, \frac{h(\bm\tbeta_k^g+\Delta\bm\tbeta_k^g) }{h(\bm\tbeta_k^g)}\right\}$.

\item Updating $\sigma^2$: draw $\sigma^2$ from its full conditional $[\sigma^2\mid \bm \tbeta, \bm \lambda, \mathcal Y]$ which is $\mbox{Inv-Ga}(a_{\sigma^2}, b_{\sigma^2})$ where $a_{\sigma^2}= 0.001 + \frac{mn}{2}$ and $b_{\sigma^2}= 0.001 +\frac{1}{2}\sum_{i=1}^n\sum_{j=1}^m \left(y_j(\bs) - \bm x_j^\top g_{\bm \lambda}[\bm\tbeta(\bs_i)] \right)^2.$

\item Updating $\bm{\lambda}$: we sequentially update $\lambda_1,...,\lambda_p$ with M-H algorithms. Specifically, for $\lambda_k$, the full conditional $[\lambda_k\mid \lambda_{-k}, \bm\tbeta, \sigma^2, \mathcal Y]$ is 
$$ \hbar(\lambda_k) =\left[ \prod_{i=1}^n\prod_{j=1}^m \phi\left(y_{j, -k}(\bs),  x_{jk}\tbeta_k(\bs_i)I_{\lambda_k}[\tbeta_k(\bs_i)], \sigma^2\right)\right]\Pi(\lambda_k),$$
where $\Pi(\lambda_k)$ is the uniform empirical Bayes prior for $\lambda_k$ defined in the previous section. The proposal for $\lambda_k$ is generate from zero mean Gaussian fluctuations as $\lambda_k+\Delta\lambda_k$, which will be accepted with probability: $\min\left\{1, \frac{\hbar(\lambda_k+\Delta\lambda_k)}{\hbar(\lambda_k)}\right\}$.

	\item Updating $\{\bm u_{k}\}_{k=1}^p$: we sequentially update $\bm u_1,...,\bm u_p$ by drawing from their full conditionals $[\bm u_k\mid \bm \tbeta, \tau_k^2]$. Specifically, we update $\bm u_k$ by drawing from $N \left(\bm \mu_{\bm u_k}, \bm \Sigma_{\bm u_k}\right)$ where $
		\bm \mu_{\bm u_k}  =  \bm\Sigma_{\bm u_k}\left(\theta^{-2}\sum_{g=1}^G\bvarphi_g^\top K_g^{-1}\bm \tbeta_k^g\right)$ and $\bm \Sigma_{\bm u_k} =  \left( \sum_{g=1}^G\theta^{-2}\bvarphi_g^\top K^{-1}_g\bvarphi_g+ \tau_k^{-2}\bm Z^{-1}\right)^{-1}$, with $\bm Z=\mbox{diag}(\zeta_1,...,\zeta_L)$.

	
	\item Updating $\{\tau_{k}^2\}_{k=1}^p$:  we sequentially update $\tau_1^2,...,\tau_p^2$ by drawing from their full conditionals $[\tau_k^2\mid \bm u_k]$. Specifically, we update $\tau_k^2$ by drawing from $\mbox{Inv-Ga}(a_{\tau_k^2}, b_{\tau_k^2})$ where $a_{\tau_k^2}= 0.001 +\frac{L}{2}$ and $b_{\tau_k^2}= 0.001+\frac{1}{2}\bm u_k^\top\bm Z^{-1}\bm u_k.$

    \item Updating the spatial range parameter $b$ within the SE kernel: this parameter can be updated by discretization. Specifically, within a reasonable range of $b$, we can calculate and store the dictionaries of $\bm \varphi_l$ and $\zeta_l$, the kernel expansion results, with regard to each discrete values of $b$ on a grid basis. Then we can update $b$ based on grid search within each MCMC iteration.

\end{itemize}

\end{appendices}

\bibliographystyle{asa}
{\setstretch{1.0}\bibliography{bib}}

\end{document}